\documentclass[12pt,a4paper]{article}
\usepackage{amsthm}
\usepackage{amsmath}
\usepackage{natbib}
\usepackage[colorlinks,citecolor=blue,urlcolor=blue,filecolor=blue,backref=page]{hyperref}
\usepackage{graphicx}
\usepackage{xcolor}
\usepackage{algorithm}
\usepackage{adjustbox}
\usepackage[noend]{algpseudocode}
\usepackage{natbib}
\usepackage{bm, dsfont}
\usepackage{amsfonts}
\newtheorem{theorem}{Theorem}[section]
\usepackage[flushleft]{threeparttable}

\usepackage[paperheight=12in,paperwidth=10in]{geometry}

\newtheorem*{theorem*}{Theorem}

\newcommand\independent{\protect\mathpalette{\protect\independenT}{\perp}}
\def\independenT#1#2{\mathrel{\rlap{$#1#2$}\mkern2mu{#1#2}}}

\DeclareMathOperator{\logit}{logit}

\begin{document}
	
	\title{\textbf{\Large Analysing the causal effect of London cycle superhighways on traffic congestion}}
	
	\date{}
	\maketitle
	\author{
		\begin{center}
			\vskip -1cm
			
	\textbf{Prajamitra Bhuyan$^{\ast, \dag}$, Emma J. McCoy$^{\ast, \dag}$,  Haojie Li$^{\S}$, Daniel J. Graham$^{\ddagger,\dag}$}\\
	$^{\ast}$Department of Mathematics, Imperial College London, United Kingdom\\
	$^{\dag}$The Alan Turing Institute, United Kingdom\\
	$^{\S}$School of Transportation, Southeast University, China\\
	$^{\ddagger}$Department of Civil and Environmental Engineering, Imperial College London, United Kingdom
			
		\end{center}
	}
	
	\begin{abstract}
Transport operators have a range of intervention options available to improve or enhance their networks. Such interventions are often made in the absence of sound evidence on resulting outcomes. Cycling superhighways were promoted as a sustainable and healthy travel mode, one of the aims of which was to reduce traffic congestion. Estimating the impacts that cycle superhighways have on congestion is complicated due to the non-random assignment of such intervention over the transport network. In this paper, we analyse the causal effect of cycle superhighways utilising pre-intervention and post-intervention information on traffic and road characteristics along with socio-economic factors. We propose a modeling framework based on the propensity score and outcome regression model. The method is also extended to the doubly robust set-up. Simulation results show the superiority of the performance of the proposed method over existing competitors. The method is applied to analyse a real dataset on the London transport network. The methodology proposed can assist in effective decision making to improve network performance.
\end{abstract}
	
	
\section{Introduction}\label{intro}
The transport network acts as a lifeline for metropolitan cities across the world. Intelligent transportation systems can revolutionize  traffic management and results in significant improvements in people's mobility. They can offer an integrated approach to infrastructure development and traffic-mobility management. The absence of well functioning commuting channels can have a strongly negative impact on those residing in urban areas. In the last couple of decades, metropolitan areas in both developed and developing countries, have been affected by increasing traffic congestion and several other problems such as poor air quality from pollution. Most air pollution in cities can be attributed to road transport and domestic and commercial heating systems.  In addition to the negative impacts on mobility and air quality, previous studies indicate that severe congestion has a negative impact on GDP and an efficient transport system significantly improves the city's economic competitiveness \citep{Cost, Growth}.




Effective design and  management of the transport network has a significant impact on the quality of life in smart cities. In general, network interventions (i.e. treatment) are a widely used measure to control high-consequence events world-wide. But often such interventions are made in the absence of statistical evidence on the resulting outcomes.  Consequently, it is common to find situations in which interventions fail to deliver their intended consequences and in which transport networks perform poorly in relation to traffic flow, speed, capacity utilisation, safety, and economic and environmental impacts. Such interventions often have unintended negative consequences.  Due to the complex nature of transport networks it is difficult to disentangle drivers of good performance and identify the factors underpinning network failure. Furthermore, it is not easy to quantify how interventions impact on system performance because transport interventions are typically targeted to address specific network problems, and are therefore non-randomly assigned. The key consequence of this non-random treatment assignment is the possibility that the effect of the treatment is `confounded' if the treated and control units differ systematically with respect to several characteristics which may affect the outcome of interest. 

In recent years, cycling has been promoted as a healthy and sustainable mode of transport, with the additional benefits of reducing traffic congestion; frequency of road accidents; and emissions from vehicle exhausts. Schemes to promote cycling have been deemed effective and the number of cyclists has increased rapidly in major European cities, including Copenhagen, Amsterdam and London. Recent reports suggest a $160\%$ increase in daily journeys in Greater London over a period of ten years from 2004 to 2014 \citep{TFL}. The Mayor of London target to achieve 400\% increase in cycling by 2026 and several policy decisions to facilitate cycling including the Cycle Superhighways (CS), Santander Cycles and Biking Boroughs have already been implemented \citep{Mayor}. The CS are 1.5 meter wide barrier-free cycling-paths designed to connect outer London to central London. The blue surfacing on CS distinguishes them from the existing cycle-paths in London (See Figure \ref{CS}). The CS routes are designed to provide adequate spatial capacity for existing cyclists and potential future commuters who adopt cycling as a mode of transport. With the aim of enabling faster and safer cycle journeys, the twelve Cycle Superhighways, were announced in 2008. As displayed in Figure \ref{Route}, these routes were designed to radiate from the city center based on the clock face layout. In July 2010, the first two pilot routes, CS3 and CS7, were inaugurated. As reported by \citet{TFL11}, in the first year, cycling has increased by $83\%$ along CS3 and $46\%$ along CS7. A new East-West route was introduced to replace CS10, while CS6 and CS12 have been cancelled. As of the end of 2015, only four routes are in operation, namely CS2 (Stratford to Aldgate); CS3 (Barking to Tower Gateway); CS7 (Merton to the City); and CS8 (Wandsworth to Westminster). Due to the lack of adequate data, the effects of CS on traffic congestion are not evaluated in the report by \citet{TFL11}. Since their introduction, there has been considerable debate about the effects of CS on road traffic congestion \citep{Guardian, Evening}. The quantification of the effects of CS on traffic congestion is a complex problem due to the intricate nature of the transport-network, and various traffic and socioeconomic characteristics may act as confounders. For example, urban areas with high population density and  narrow roads are expected to have more congestion compared to the outskirts. Similarly, the number of public road transport stops will have an effect on pedestrian activities resulting in changes in traffic congestion.

\begin{figure}[htp]	
	\centering
	\includegraphics[scale=0.2]{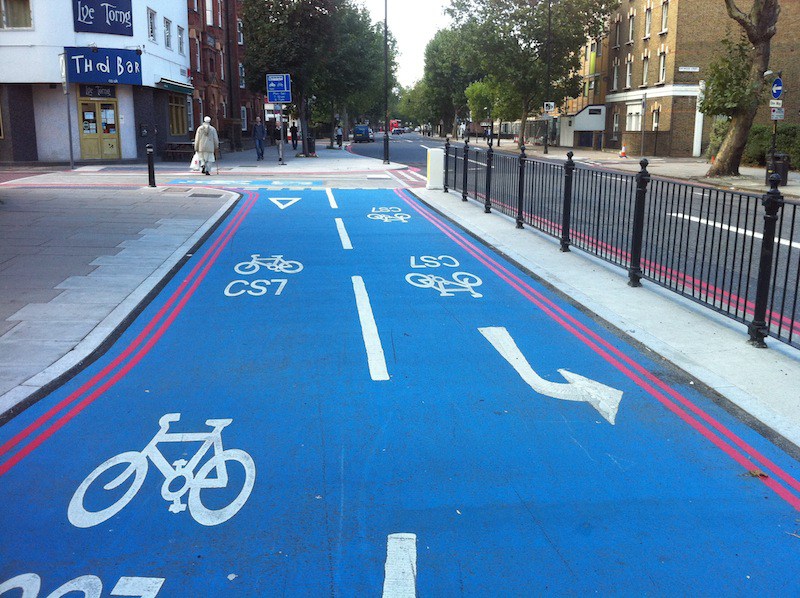}
	\includegraphics[scale=0.2]{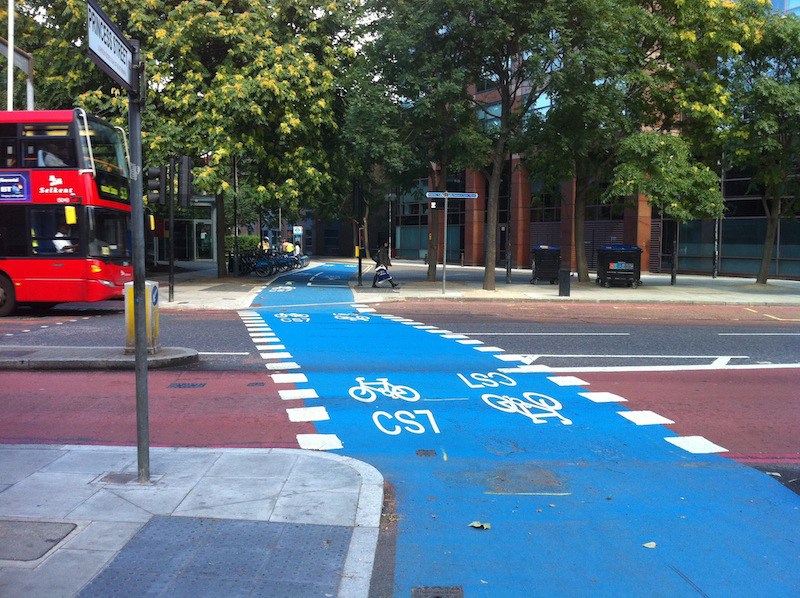}
	
	\includegraphics[scale=0.2]{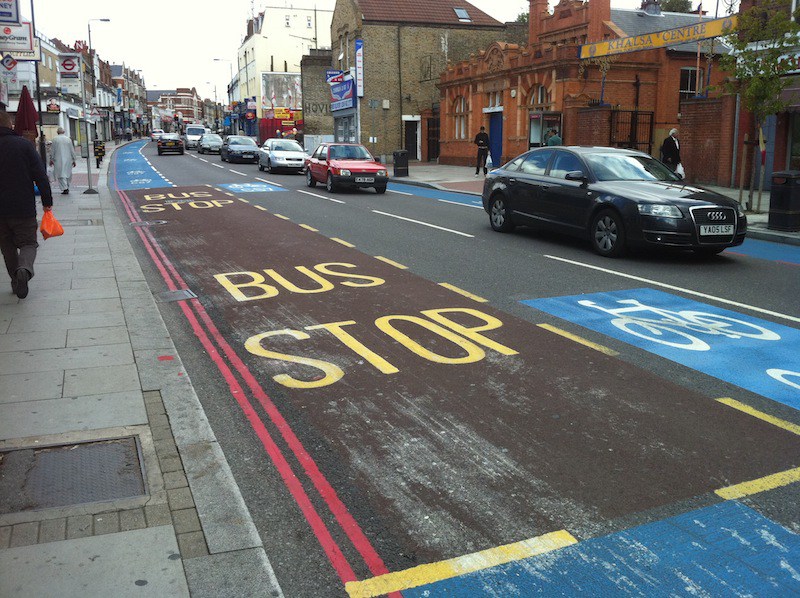}
	\includegraphics[scale=0.2]{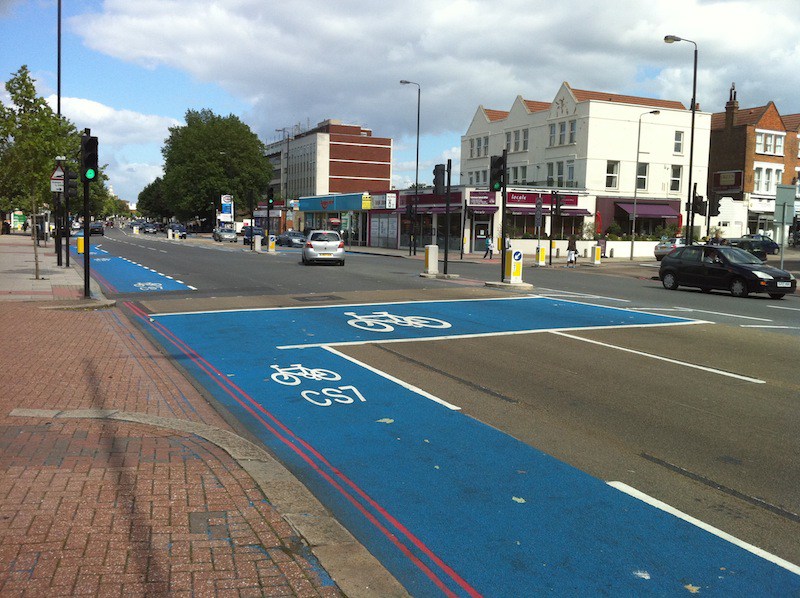}
	\caption{Cycling Superhighways}
	\label{CS}
\end{figure}

\begin{figure}[htp]	
	\centering
	\includegraphics[scale=0.4]{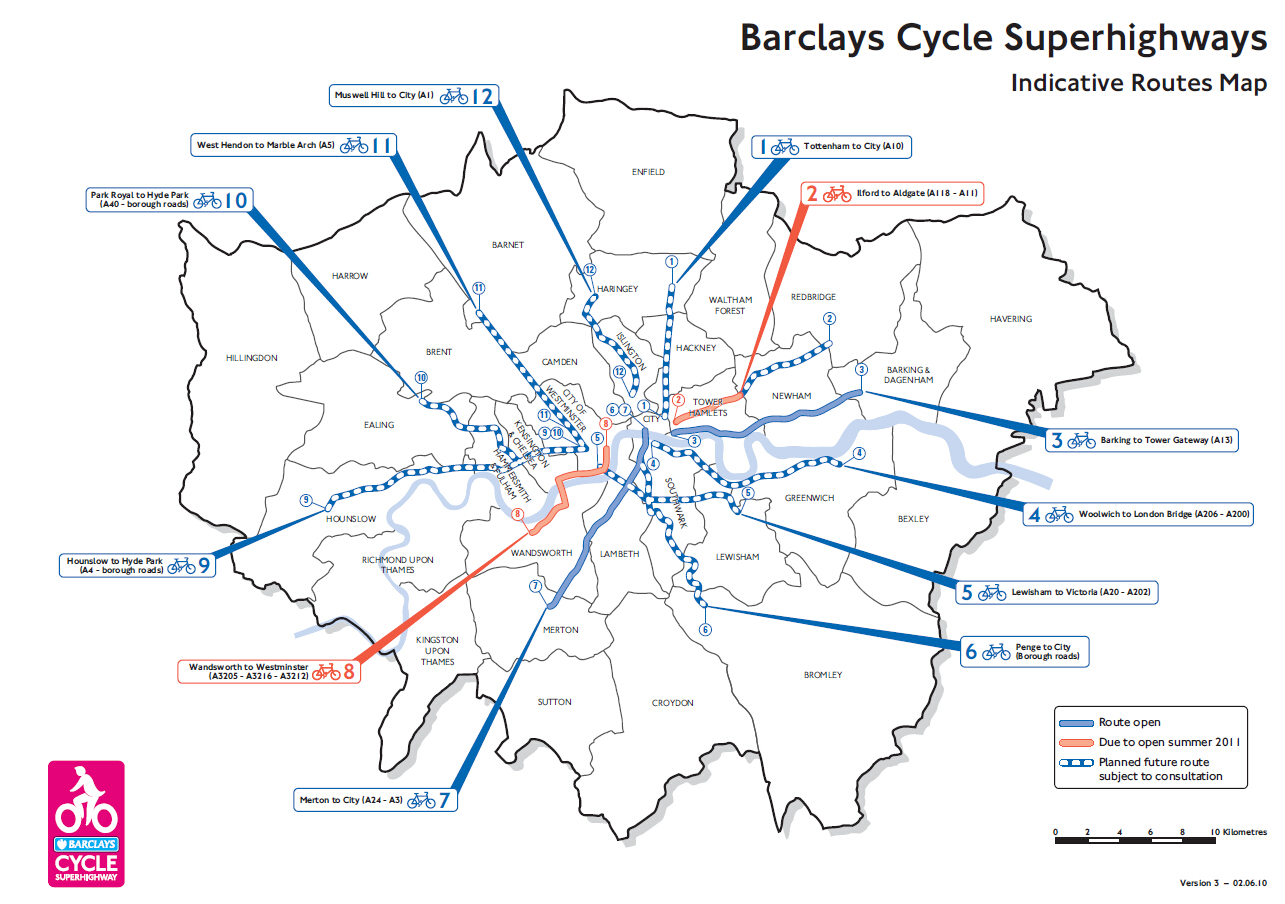}
	\caption{Route map of Cycling Superhighways}
	\label{Route}
\end{figure}

Congestion arises when the volume of traffic increases to a level that causes traffic speeds to fall below the free-flow (i.e. maximum) speed. There is a direct relationship between flow and traffic speed, both of which have an impact on the levels of congestion \citep{Link1, Link2}.
In order to investigate the impact of the introduction of Cycle Superhighways on congestion in London, we present a causal analysis of the links between traffic flow and speed and Cycle Superhighways. Our objective is to provide a modeling framework to obtain robust estimates of the causal quantities addressing issues of confounding and longitudinal dependence based on pre-intervention and post-intervention data. Using generalized linear mixed models (GLMMs), we incorporate time-invariant and time-varying covariates to adjust for potential sources of confounding and bias from longitudinal dependence by using random effect parameters. We also propose a novel estimator to deal with unknown interactions between time and covariates. The proposed methods are extended to a doubly robust set-up. First, we describes the existing models and methods used in the traditional potential outcome framework in Section \ref{frame}. In Section \ref{Prop}, we propose new methods motivated by real-life data and explain the advantages compared to the existing methods available in the literature. Simulation studies are performed for assessing the effectiveness of the proposed methods and the results are summarized in Section \ref{Sim}. In Section \ref{Data}, results from real data analysis on London Cycle Superhighways are discussed. We summarise the key findings and conclude with some discussion on future research in Section \ref{Conc}.

\section{Potential Outcome Framework}\label{frame}
In traditional causal inference problems, the primary interest is to estimate the average treatment effect (ATE) based on available data as realisations of a random vector, $Z_i = (Y_i,D_i,X_i)$, $i = 1,\ldots, n$, where $Y_i$ denotes a response, $D_i$ the exposure (or treatment status), and $X_i$ a vector of confounders or covariates for the ith unit. The treatment can be binary, multi-valued or continuous but essentially it is not assigned randomly. In this set-up, the simple comparisons of mean responses across different treatment groups lead to biased estimates and may not reveal a ‘causal’ effect of the treatment due to potential confounding. Confounding can be addressed if the vector of covariates $X_i$ is sufficient to ensure unconfoundedness, or conditional independence of potential responses and treatment assignment. In the context of binary treatments, the conditional independence assumption requires that $(Y_i(0), Y_i(1)) \independent I(D_i)|X_i$, where $I(D_i)$ is the indicator function for receiving the treatment and $Y (1)$ and $Y (0)$ indicate potential outcomes under treated or control status, respectively. An additional requirement for valid causal inference is that, conditional on covariates $X_i$, the probability of assignment to treatment is strictly positive for every combination of values of the covariates, i.e. $0<P\left[I(D_i)=1| X_{i}=x\right]<1$ for all $x$. See \citet{Bum} for more details. The main interest is to estimate the average treatment effect $\tau = \mathbb{E}[Y_i(1)]- \mathbb{E}[Y_i(0)]$, which measures the difference in expected outcomes under treatment and control status. In many applications, it is also of interest to understand what would be the effect (on average) of applying (vs. withholding) treatment among all units who are prescribed such interventions rather than average effect of applying, versus withholding, treatments among all units. This causal quantity is formally known as average treatment effect on treated (ATT) and defined as $\kappa=\mathbb{E}[Y_i(1)|D_i=1]- \mathbb{E}[Y_i(0)|D_i=1]$.
When there exists no heterogeneity of the treatment effect by covariates (i.e., no interactions effects of treatment and confounder), the ATE and ATT are identical if effects are additive. However, when treatment heterogeneity exists, the ATE and ATT differ \citep{Erica}.

Several estimators for the ATE are studied in the literature under the assumption that the covariate vector is sufficient to ensure the independence  of  potential  responses  and  treatment  assignment, see for example, \cite{HernanRobins}. Three estimators are of particular interest here. First, we model the expected response given the covariates and treatment, $\mathbb{E}[Y_i|X_i,D_i]$, using an outcome regression (OR) model $\Psi^{-1}\left\{m(X_i,D_i; \xi)\right\}$, for known link function $\Psi$, regression function $m(\cdot)$, and unknown parameter vector $\xi$. If the OR model is correctly specified, then a consistent estimate of the ATE is given by 
\begin{eqnarray}\label{OR}
\hat{\tau}_{OR}=\frac{1}{n}\sum_{i=1}^{n}\left[\Psi^{-1}\left\{ m(D_{i}=1,X_{i};\hat{\xi})\right\}- \Psi^{-1}\left\{m(D_{i}=0,X_{i};\hat{\xi})\right\} \right]\nonumber,
\end{eqnarray}
where $\hat{\xi}$ is iteratively reweighted least squares estimate \citep{Beng}. Similarly, the estimate of ATT is given by 
\begin{eqnarray}\label{OR_ATT}
\hat{\kappa}_{OR}=\frac{1}{\sum_{i=1}^{n}D_i}\sum_{i=1}^{n}D_{i}\left[\Psi^{-1}\left\{ m(D_{i}=1,X_{i};\hat{\xi})\right\}- \Psi^{-1}\left\{m(D_{i}=0,X_{i};\hat{\xi})\right\} \right]\nonumber.
\end{eqnarray} Second, we could assume a model for $f_{D|X}(d|x)$, the conditional density of the treatment given the covariates and use this model to estimate propensity scores (PS), which we denote $\pi(D_i|X_i; \hat{\alpha})$ with parameter vector $\alpha$.  A PS weighted estimator of the form attributed to \citet{HT} is then given by
\begin{eqnarray}\label{PS}
\hat{\tau}_{IPW}=\frac{1}{n}\sum_{i=1}^{n}\left[\frac{I(D_{i})Y_{i}}{\pi(D_{i}|X_{i}, \hat{\alpha})}-\frac{\left[1-I(D_{i})\right]Y_{i}}{1-\pi(D_{i}|X_{i}, \hat{\alpha})} \right],\nonumber
\end{eqnarray}
which can be used to estimate the ATE consistently \citep{Hirano_2003, Dav}. In the econometrics literature, there has been considerable attention devoted to the estimation of the ATT via matching. Matching is appealing and an intuitive approach, but matching estimators also have several unattractive statistical properties \citep{Abadie_2006}. The PS weighted estimator for ATT has been proposed previously by \citet{Hahn_1998,Hirano_2003}. \citet{Hahn_1998} considered nonparametric imputation methods to estimate the counterfactual outcomes under both treatment conditions and then marginalize these differences among the treated individuals using inverse probability of treatment weights. \citet{Hirano_2003} proposed an efficient estimate of ATT based on inverse PS weights which can be interpreted as an empirical likelihood estimator. \citet{Erica} provided a new and intuitive understanding of the usual weighted estimator from the perspective of importance sampling and proposed a consistent estimator
of the ATT, given by 
\begin{eqnarray}\label{PS_ATT}
\hat{\kappa}_{IPW}=\frac{1}{\sum_{i=1}^{n}D_i}\sum_{i=1}^{n}\left[I(D_{i})-\frac{\left[1-I(D_{i})\right]\pi(D_{i}|X_{i},\hat{\alpha})}{1-\pi(D_{i}|X_{i}, \hat{\alpha})}\right]Y_i.
\end{eqnarray}
Finally, we could combine an
OR and PS model and construct an estimate for the ATE and ATT. This is known as a doubly robust (DR) method in the sense that the resulting estimator is asymptotically consistent if either the OR or PS model is correctly specified. See \citet{DR} for more details.

It is important to note that the aforementioned estimators can only incorporate post-intervention observations on the response variable. In our context, pre-intervention measurements are also available, which potentially contain important information on the causal quantities of interest. One can still use the existing methods for estimation of ATEs and ATTs discarding the pre-intervention data, but the estimates are possibly less efficient compared to any estimator suitable for incorporating the entire dataset (See Section \ref{Sim} later). In the next section, we attempt to address this issue and propose new estimator for estimation of ATE and ATT based on both pre-intervention and post-intervention measurements.

\section{Proposed Methods}\label{Prop}
In our context, the response variable is observed both before and after the intervention and it depends on time-independent as well as time-dependent covariates.  The available data can be represented in a longitudinal structure as $Z_{it} = (Y_{it},D_{it},X_{it})$, $i = 1,\ldots, n$, $t=0,1$, where $t=0$ and $t=1$ denote pre-intervention and  post-intervention-period, respectively. Note that $D_{i0}=0$ for all $i=1,\ldots,n$, and  $X_{it}$ contains both time-independent and time-dependent covariates. The policy decisions are formulated based on pre-intervention condition, that is the treatment status $D_{i1}$ is only dependent on the $X_{i0}$.  

We consider the following assumptions for a valid inference on causal quantities. First, to address confounding we assume ``unconfoundedness'', or conditional independence between the response and treatment assignment given the covariates:
$(Y_{it}(0), Y_{it}(1)) \independent I(D_{i1})|(X_{i1}, X_{i0})$, where $Y_{it}(1)$ and $Y_{it}(0)$ are potential outcomes under treated and control status, respectively, at time $t$. Second, to ensure comparability across potential outcomes there must be common support, or overlap, by treatment status in the covariate distributions: $0<P\left[I(D_{i1})=1| X_{i0}=x\right]<1$ for all $x$. Third, for logical consistency there must be equivalence between the observed response under a given treatment and the potential response: $Y_{it}=I(D_{it})Y_{it}(1)+ \left[1-I(D_{it})\right]Y_{it}(0)$ for $t=0, 1$, $i=1,2,\ldots,n$. One important implication of this is that the stable unit treatment value assumption (SUTVA) (e.g., \citet{Rubin}) must hold, which requires: (i) that the outcome for each unit be independent of the treatment status of other units, or in other words, there should be no interference in treatment effects across units; and (ii) that there are no different versions of the treatment. Finally, we consider that the covariate $X_{it}$ is not affected by treatment status $D_{it}$ \citep{Daniel}.
In this set-up, our main interest is to estimate the causal effects $\tau_{0} = \mathbb{E}[Y_{i1}(1)]- \mathbb{E}[Y_{i1}(0)]$, and $\kappa_{0} = \mathbb{E}[Y_{i1}(1)|D_{i1}=1]- \mathbb{E}[Y_{i1}(0)|D_{i1}=1]$, which measures the difference in expected outcomes under treatment and control status in the post-intervention period for the entire population and treatment group, respectively.

In the literature of econometrics, a classical approach is to use the difference-in-difference (DID) method to adjust the effect of confounding based on pre-intervention and post-intervention data \citep{Lechner_2010}. This method is popularly used to estimate the causal effect of policy decisions in the field of health economics \citep{Health1, Health2}, epidemiology \citep{Epi}, market economics \citep{Market}, and various other allied fields. \citet{Lechner_2010} provided a detailed survey of the literature on the DID estimation strategy and discussed identification issues from the treatment effects perspective. In many economic applications, the identification conditions for ATE are unlikely to be plausible. Thus, empirical papers using DID mainly focused on the ATT and do not discuss estimation of the ATE. The DID method is easy to compute and intuitively appealing, however, it assumes the `parallel trend' assumption holds, that is, that the difference between the ‘treatment’ and ‘control’ group is constant over time in the absence of treatment.  Under the parallel trend assumption, the DID estimator is given by
\begin{eqnarray}\label{DID}
\hat{\kappa}_{DID}&=&\sum_{i=1}^{n}\left[\frac{I(D_{i1})Y_{i1}}{\sum_{i=1}^{n}I(D_{i1})}-\frac{\left[1-I(D_{i1})\right]Y_{i1}}{\sum_{i=1}^{n}(1-I(D_{i1}))} \right]\nonumber\\
&&-\sum_{i=1}^{n}\left[\frac{I(D_{i1})Y_{i0}}{\sum_{i=1}^{n}I(D_{i1})}-\frac{\left[1-I(D_{i1})\right]Y_{i0}}{\sum_{i=1}^{n}(1-I(D_{i1}))} \right].\nonumber
\end{eqnarray}
The parallel trend assumption may be implausible if pre-treatment characteristics that are thought to be associated with the response variable are unbalanced between the treated and the control group \citep{Parallel}. One popular method for reducing bias when outcome trends are not parallel is to first match treatment and control observations on pre-treatment outcomes before applying difference-in-differences on the matched sample. However, some studies suggest that this approach does not always eliminate or reduce bias \citep{Lindner_2018}.
To address the same issue, covariate adjusted DID estimators are proposed in the literature. \citet{DID} proposed a way to accommodate  covariates in the DID estimator based on a nonparametric conditional difference-in-differences extension of the method of matching. \citet{Parallel} exploited the same identification criteria considered by \citet{DID}, and proposed a flexible approach for the case in which differences in observed characteristics create non-parallel outcome dynamics between treated and controls. It also allows flexible effects of covariates on outcome variables to incorporate heterogeneous treatment effects. \citet{Callaway} extended DID methods for multiple time periods with variation in treatment timing and considered clustered bootstrapped standard errors to account for the serially correlated data. The possibility that the
treatments and potential outcomes of a group may be correlated over time is also considered by \citet{Chaisemartin_2020}. Recently, \citet{SantAnna_2020} proposed a doubly robust DID estimator which is consistent if either a propensity score or outcome regression model are correctly specified.

In the following subsections, we propose novel modeling approaches to address these key estimation issues. We provide a modeling framework to obtain estimates of the ATE and ATT addressing issues of confounding, longitudinal dependence, and model misspecification. The proposed methods flexibly incorporate time-invariant  and  time-varying  covariates and deal with unknown time-interaction and heterogeneous treatment effects. As mentioned before, the proposed methods rely on the `unconfoundedness' assumption motivated by the availability of various confounding factors and associated domain knowledge for the real application at hand. Unlike cross-sectional designs, it is not required for a difference-in-differences set-up to have similar baseline means in the outcome or other covariates for the treatment and control groups. As a result, a confounder in a difference-in-differences study is any variable related to both treatment assignment and the change in the outcome over time \citep{Daw_2018}. In our context, time-varying covariates are serially correlated and they are related to both outcome trend and its level. Similarly, interactions between time and other covariates act as a time-varying confounder and can be taken into consideration.
Also, a confounder is any time-invariant covariate related to both treatment assignment and the outcome level at the post-intervention period.

\subsection{A Generalized Mixed Model Approach}\label{GLMM}
In this section we develop an estimator of the ATE and ATT which can account for the situation where the response variable depends on both time-independent and time-dependent covariates as well as exposure and, additionally, where the pre-intervention and post-intervention responses of the same unit may be serially correlated. To estimate the ATE and ATT incorporating these features, we consider the following GLLM model:
\begin{eqnarray}\label{GLLM}
\Psi\left[\mathbb{E}\left(Y_{it}|D_{it},\boldsymbol{X_{it}}, \boldsymbol{Z_{it}}, \boldsymbol{u_{i}}\right)\right]= \gamma t + \beta D_{it} + \boldsymbol{X_{it}}^{T}\boldsymbol{\theta} + \boldsymbol{Z_{it}}^{T}\boldsymbol{u_{i}}, i=1,\ldots,n, t=0,1,
\end{eqnarray}
where $\Psi$ is some known link function, and $\gamma t + \beta D_{it}+\boldsymbol{X_{it}}^{T}\boldsymbol{\theta}$ is the fixed effects with time effect $\gamma$, treatment effect $\beta$, parameter vector $\boldsymbol{\theta}$ corresponding to the design vector $\boldsymbol{X_{it}}$, and $\boldsymbol{Z_{it}}$ is the design vector for the random effects $\boldsymbol{u_{i}}\sim N(\boldsymbol{0},\boldsymbol{G})$ with $\boldsymbol{G}$ being a positive definite matrix. Note that the design vector $\boldsymbol{X_{it}}$, and $\boldsymbol{Z_{it}}$ may also include the interactions of covariates with time $t$ and treatment $D_{it}$ to account for the heteroginity with respect to time and treatment, respectively. For the linear predictor $\eta_{it} =\gamma t + \beta D_{it} +\boldsymbol{X_{it}}^{T}\boldsymbol{\theta} + \boldsymbol{Z_{it}}^{T}\boldsymbol{u_{i}}$, the conditional expectation is $\mu_{it} =\mathbb{E}\left[Y_{it} |\eta_{it}\right]$ and the conditional variance is var$\left(Y_{it} |\eta_{it}\right)=\phi V(\mu_{it})$, where $V(\mu_{it})$ is the variance function and $\phi$ is a dispersion parameter. Note that the aforementioned model, given by (\ref{GLLM}), can be used to incorporate observations on multiple time-periods before and after the intervention.


For our case study, the relationship between the response and treatment is likely to be confounded in the sense that both the response (i.e. traffic flow or speed) and the treatment (i.e. cycle superhighways) could depend on a set of pre-intervention characteristics. Some of the characteristics may evolve over time that changes its distribution in the post-intervention period being serially correlated with the pre-intervention distribution. Moreover, several other factors in the post-intervention period could result in a  spurious association between response and treatment. To address these issues, we include both time-independent and time-dependent covariates within the design matrix $\boldsymbol{X_{it}}$. Also, the proposed model accounts for linear or non-linear time trend affecting the response. The response variable measured over different time points are serially correlated and the random effects $\boldsymbol{u_{i}}$ account for the same and other sources of unobserved heterogeneity \citep[Ch-7]{Diggle}. 
In real applications, normality of random effects is typically assumed for computational convenience, but it may be misspecified. However, theoretical and simulation studies indicate that most aspects of statistical inference, including estimation of covariate effects and estimation of the random effects variance, based on maximum likelihood methods are highly robust to
this assumption \citep{McCulloch_2011}.

The estimate of ATEs based on model (\ref{GLLM}) involves predictions for both the treatment statuses at the post-intervention period. For non-linear link functions, population-averaged expectations cannot be obtained by simply plugging the mean of the random effects in the expression for conditional expectation \citep{RE}. Using the double-expectation rule, the estimate of $\tau_{0}$ is obtained as
\begin{align*}
	\hat{\tau}_{GLMM}=\frac{1}{n}\sum_{i=1}^{n}\int_{-\infty}^{\infty}\left[\Psi^{-1}\left\{ \hat{\gamma}  + \hat{\beta}  + \boldsymbol{X_{i1}}^{T}\boldsymbol{\hat{\theta}} + \boldsymbol{Z_{i1}}^{T}\boldsymbol{u_{i}}\right\} \right.\\  \left. -\Psi^{-1}\left\{\hat{\gamma}  +  \boldsymbol{X_{i1}}^{T}\boldsymbol{\hat{\theta}} + \boldsymbol{Z_{i1}}^{T}\boldsymbol{u_{i}}\right\} \right] \boldsymbol{\varphi}(\boldsymbol{u_{i}}; \boldsymbol{\hat{G}})d\boldsymbol{u_{i}}\nonumber,
\end{align*}
where $\boldsymbol{\varphi}(\boldsymbol{u_{i}}; \boldsymbol{\hat{G}})$ is the density function of the random effects $\boldsymbol{u_{i}}$ with estimated covariance matrix $\boldsymbol{\hat{G}}$. The integral that is involved in the expression of $\hat{\tau}_{GLMM}$ must generally be evaluated numerically or by simulation. Similarly, the estimate of $\kappa_{0}$ is given by
\begin{align*}
	\hat{\kappa}_{GLMM}=\frac{1}{\sum_{i=1}^{n}D_{i1}}\sum_{i=1}^{n}\int_{-\infty}^{\infty}D_{i1}\left[\Psi^{-1}\left\{ \hat{\gamma}  + \hat{\beta}  + \boldsymbol{X_{i1}}^{T}\boldsymbol{\hat{\theta}} + \boldsymbol{Z_{i1}}^{T}\boldsymbol{u_{i}}\right\} \right.\\  \left. -\Psi^{-1}\left\{\hat{\gamma}  +  \boldsymbol{X_{i1}}^{T}\boldsymbol{\hat{\theta}} + \boldsymbol{Z_{i1}}^{T}\boldsymbol{u_{i}}\right\} \right] \boldsymbol{\varphi}(\boldsymbol{u_{i}}; \boldsymbol{\hat{G}})d\boldsymbol{u_{i}}\nonumber.
\end{align*}

\subsection{Inverse Propensity Weighted Difference-in-Difference}\label{ITWDID}
Historically, researchers in various fields of application have used regression based methods to measure the differences between the treated and control group. More recently, PS based methods have become increasingly popular to eliminate the effects of confounding present in observational data.  PS based models have several advantages, it is simpler to determine the adequacy of a PS model than to assess whether the regression model reasonably specifies the relationship between the exposure and covariates.  Moreover, standard goodness-of-fit tests fail to identify whether the fitted regression model has successfully accounted for the systematic differences between treated and control groups for the estimation of ATE \citep{PS}. In our context, the response variable not only depends on the pre-intervention confounders but is also affected by various factors, possibly unmeasured, in the post-intervention period. In particular, regression based method provide biased estimates in the presence of unknown interactions between time and covariates (See Section \ref{Sim} later). In order to avoid the difficulties that arise in regression based approaches, we propose the following difference-in-difference estimator for ATE based on an inverse propensity weighting:
\begin{eqnarray}\label{IPWDID}
\hat{\tau}_{IPWDID}&=&\frac{1}{n}\sum_{i=1}^{n}\left[\frac{I(D_{i1})Y_{i1}}{\pi(D_{i1}|\boldsymbol{X_{i0}}, \hat{\alpha})}-\frac{\left[1-I(D_{i1})\right]Y_{i1}}{1-\pi(D_{i1}|\boldsymbol{X_{i0}}, \hat{\alpha})} \right]\nonumber\\
&&-\frac{1}{n}\sum_{i=1}^{n}\left[\frac{I(D_{i1})Y_{i0}}{\pi(D_{i1}|\boldsymbol{X_{i0}}, \hat{\alpha})}-\frac{\left[1-I(D_{i1})\right]Y_{i0}}{1-\pi(D_{i1}|\boldsymbol{X_{i0}}, \hat{\alpha})} \right],
\end{eqnarray}
where $\pi(D_{i1}|\boldsymbol{X_{i0}}, \hat{\alpha})$ is the estimated PS based on pre-intervention covariate vector $\boldsymbol{X_{i0}}$. Logit and probit models are widely used and perform reasonably well to estimate the PS \citep{PS}, however, one can use the generalized additive model (GAM), or machine learning techniques such as random forests, neural network, etc., to avoid the bias induced by a misspecified parametric model \citep{Westreich_2010, Lee_2010}. Here also, the proposed estimate $\hat{\tau}_{IPWDID}$ can also be generalized to incorporate observations on multiple time-periods. For this purpose, one may consider different PS models based on multiple pre-intervention observations, as discussed in \citet{Lindner_2018}, and replacing $Y_{i0}$ and $Y_{i1}$ in (\ref{IPWDID}) by its averages for the pre-intervention and post-intervention periods, respectively. The inverse propensity score difference-in-difference estimate $\hat{\tau}_{IPWDID}$ is consistent under the condition that the PS model $\pi(D_{i1}|\boldsymbol{X_{i0}}, \hat{\alpha})$ is correctly specified, and the expected potential outcome under treatment and control status are equal in the pre-intervention period (i.e. $\mathbb{E}[Y_{i0}(1)]=\mathbb{E}[Y_{i0}(0)]$). The proof is outlined in the Appendix. 

To estimate the ATT, we consider the estimate proposed by \citet{Erica}, given by (\ref{PS_ATT}), and modify the same to accommodate pre-intervention and post-intervention measurements on the response variable. The proposed difference-in-difference estimator based on an inverse propensity weighting is given by
\begin{eqnarray}\label{IPWDID_ATT}
\hat{\kappa}_{IPWDID}&=&\frac{1}{\sum_{i=1}^{n}D_{i1}}\sum_{i=1}^{n}\left[I(D_{i1})-\frac{\left[1-I(D_{i1})\right]\pi(D_{i1}|X_{i},\hat{\alpha})}{1-\pi(D_{i1}|X_{i}, \hat{\alpha})}\right]Y_{i1}\nonumber\\
&&-\frac{1}{\sum_{i=1}^{n}D_{i1}}\sum_{i=1}^{n}\left[I(D_{i1})-\frac{\left[1-I(D_{i1})\right]\pi(D_{i1}|X_{i},\hat{\alpha})}{1-\pi(D_{i1}|X_{i}, \hat{\alpha})}\right]Y_{i0}\nonumber.
\end{eqnarray}
A similar estimator is also proposed by \citet{Parallel}. Under the same assumptions, as considered for the estimation of ATE, $\hat{\kappa}_{IPWDID}$ is consistent for the ATT.

\subsection{Doubly Robust Method}\label{DR}
We have proposed two different methods that adjust the parameter estimates in the case where covariates may be related both to the response and treatment assignment mechanism. In Subsection \ref{GLMM}, we model the relationships between the covariates and the response and use those relationships to predict for both the treatment statuses and obtain estimate of ATE and ATT. Another approach, discussed in Subsection \ref{ITWDID}, is to model the probabilities of treatment assignment given the covariates and incorporate them into a weighted difference-in-difference estimate. However, with observational data, one can never be sure that a model for the treatment assignment mechanism or an outcome regression model is correct, an alternative approach is to develop a  doubly-robust (DR) estimator. Several DR estimation methodologies are proposed in the literature \citep{DR}. In this paper, we propose a DR approach that is close in spirit to the approach considered by \citet{SantAnna_2020} but the modeling technique and associated inference are different. In particular, we extend the model given by (\ref{GLLM}) through the augmented regression method as
\begin{eqnarray}\label{GLMMDR}
\Psi\left[\mathbb{E}\left(Y_{it}|D_{it},\boldsymbol{X_{it}}, \boldsymbol{Z_{it}}, \boldsymbol{u_{i}}\right)\right]= \gamma t + \beta D_{it} + \boldsymbol{X_{it}}^{T}\boldsymbol{\theta} +  \boldsymbol{Z_{it}}^{T}\boldsymbol{u_{i}} + \zeta h(\pi(D_{i1}|\boldsymbol{X_{i0}}, \hat{\alpha})),
\end{eqnarray}
where $h(\cdot)$ is a suitably chosen parametric function and $\zeta$ is the associated coefficient. \citet{Rot} considered inverse PS as an additional covariate, i.e. $h(z)=1/z$, and \citet{Beng} proposed the so called `clever covariate' which is also a function of the inverse PS and treatment status. \citet{DR} compared the performance of DR estimates with various choices of $h(\cdot)$ and found that these two choices perform very poorly when some of the estimated propensities are small. The best performance is achieved by choosing $h(\cdot)$ as a step-wise constant function with discontinuities at the sample quantiles of $\pi(D_{i1}|\boldsymbol{X_{i0}}, \hat{\alpha})$. In other words, one can simply coarse classify the PS into some suitable number of categories and create dummy indicators to augment the regression model. With this choice, the model (\ref{GLMMDR}) can be rewritten as  
\begin{eqnarray}\label{GLMMDRR}
\Psi\left[\mathbb{E}\left(Y_{it}|D_{it},\boldsymbol{X_{it}}, \boldsymbol{Z_{it}}, \boldsymbol{u_{i}}\right)\right]= \gamma t + \beta D_{it} + \boldsymbol{X_{it}}^{T}\boldsymbol{\theta} +  \boldsymbol{Z_{it}}^{T}\boldsymbol{u_{i}} + \boldsymbol{W_i}^{T}\boldsymbol{\zeta} ,
\end{eqnarray}
where $\boldsymbol{W_{i}}$ is the vector of dummy variables to indicate the category based on the coarse classification of $\pi(D_{i1}|\boldsymbol{X_{i0}}, \hat{\alpha})$ and $\boldsymbol\zeta$ is the associated parameter vector. Here also, the augmented regression model, given by (\ref{GLMMDRR}), can be used to incorporate observations on multiple time-periods. As discussed before, one may consider different PS models based on multiple pre-intervention observations, as discussed in \citet{Lindner_2018}, to obtain the augmented covariate. The DR estimate of $\tau_{0}$ based on model (\ref{GLMMDRR}) is given by
\begin{align*}
	\hat{\tau}_{DRGLMM}=\frac{1}{n}\sum_{i=1}^{n}\int_{-\infty}^{\infty}\left[\Psi^{-1}\left\{ \hat{\gamma}  + \hat{\beta}  + \boldsymbol{X_{i1}}^{T}\boldsymbol{\hat{\theta}} + \boldsymbol{Z_{i1}}^{T}\boldsymbol{u_{i}}+\boldsymbol{W_i}^{T}\boldsymbol{\hat{\zeta}}\right\}\right.\\
	\left. - \Psi^{-1}\left\{\hat{\gamma}  +  \boldsymbol{X_{i1}}^{T}\boldsymbol{\hat{\theta}} + \boldsymbol{Z_{i1}}^{T}\boldsymbol{u_{i}}+\boldsymbol{W_i}^{T}\boldsymbol{\hat{\zeta}}\right\} \right] \boldsymbol{\varphi}(\boldsymbol{u_{i}}; \boldsymbol{\hat{G}})d\boldsymbol{u_{i}}\nonumber.
\end{align*}
Similarly, the DR estimate of $\kappa_{0}$ based on model (\ref{GLMMDRR}) is given by
\begin{align*}
	\hat{\kappa}_{DRGLMM}=\frac{1}{\sum_{i=1}^{n}D_{i1}}\sum_{i=1}^{n}\int_{-\infty}^{\infty}D_{i1}\left[\Psi^{-1}\left\{ \hat{\gamma}  + \hat{\beta}  + \boldsymbol{X_{i1}}^{T}\boldsymbol{\hat{\theta}} + \boldsymbol{Z_{i1}}^{T}\boldsymbol{u_{i}}+\boldsymbol{W_i}^{T}\boldsymbol{\hat{\zeta}}\right\}\right.\\
	\left. - \Psi^{-1}\left\{\hat{\gamma}  +  \boldsymbol{X_{i1}}^{T}\boldsymbol{\hat{\theta}} + \boldsymbol{Z_{i1}}^{T}\boldsymbol{u_{i}}+\boldsymbol{W_i}^{T}\boldsymbol{\hat{\zeta}}\right\} \right] \boldsymbol{\varphi}(\boldsymbol{u_{i}}; \boldsymbol{\hat{G}})d\boldsymbol{u_{i}}\nonumber.
\end{align*}
The performance of $\hat{\tau}_{DRGLLM}$ and $\hat{\kappa}_{DRGLMM}$ may depend on the dimension of $\boldsymbol{\zeta}$. In practice, it is observed that the performance of the DR estimate is satisfactory with no more that four dummy variables \citep{DR}. The estimate of $\boldsymbol{\zeta}$ converges to $\boldsymbol{0}$ if the GLMM is correctly specified. If the PS model is correct, but the GLMM is not, the augmented regression has a bias correction property. As expected, the DR estimate does not necessarily translate into good performance when neither model is correctly specified. See \citet{Rot} for details. Also, it is important to note that the DR estimate may produce biased results in the presence of unknown interactions between time and covariates. 

\section{Simulation Study}\label{Sim}
In order to compare the performance of the proposed methods and existing competitors, we generate data comprising $250$ and $500$ units. We generate three independent covariates $X_{1it}$,  $X_{2i}$ and $V_{i}$, where $(X_{1i0},X_{1i1})'$ follows a bivariate normal with mean vector $(15,20)'$ and covariance matrix 
$\bigl(\begin{smallmatrix}
6&5.5 \\ 5.5&6
\end{smallmatrix} \bigr)$;
$X_{2i}$ follows an exponential distribution with mean $2$; and $V_{i}$ follows a normal distribution with mean $1$ and variance $1$. The covariate $X_{1it}$ varies with time while $X_{2i}$ and $V_{i}$ are time-invariant. We first consider homogeneous treatment effect and specify the following relationships between the covariates and the treatment $D_{it}$ and response $Y_{it}$:
\begin{eqnarray}\label{Model-I}
Y_{it}&=& 10 + 3t + 15 D_{it} + X_{1it} + 2X_{2i} + \log(X_{2i})+  u_{i} + \epsilon_{it},  \nonumber\\
\logit\left[ \pi(D_{i1}|X_{1i0},X_{2i},V_{i}) \right]&=& -3 + 0.2X_{1i0} + 0.1 X_{2i} + 0.3V_{i},\nonumber
\end{eqnarray}
where random effect $u_{i}$ and error $\epsilon_{ti}$ are generated from independent normal distributions with mean $0$ and variance $30$, and $20$, respectively. Thus, $X_{1it}$ and $X_{2i}$ are confounders and $V_{i}$ is a non-confounding covarite. 
Under this set-up, the following estimators are tested:
\sloppy
\begin{itemize}
	\item[1.] $\hat{\tau}_{OR}$ - correctly specified OR model based on post-intervention measurements: $\mathbb{E}\left[Y_{i1}|D_{i1},X_{1i1},X_{2i} \right]= \theta_{0}  + \beta D_{i1} + \theta_{1}X_{1i1} + \theta_{2}X_{2i} + \theta_{3}\log(X_{2i})$.
	\item[2.] $\tilde{\tau}_{OR}$ - incorrectly specified OR model based on post-intervention measurements with erroneous exclusion of the time-invariant confounder $X_{2i}$: $\mathbb{E}\left[Y_{i1}|D_{i1},X_{1i1},X_{2i} \right]= \theta_{0}  + \beta D_{i1} + \theta_{1}X_{1i1}$.
	\item[3.] $\hat{\tau}_{GLMM}$ - correctly specified GLMM model: $\mathbb{E}\left[Y_{it}|D_{i1},X_{1i1},X_{2i}, u_{i} \right]= \theta_{0} +\gamma t + \beta D_{it} + \theta_{1}X_{1it} + \theta_{2}X_{2i}+ \theta_{3}log(X_{2i})+u_{i}$.
	\item[4.] $\tilde{\tau}_{GLMM}$ - incorrectly specified GLMM model with erroneous
	exclusion of the time-invariant confounder $X_{2i}$: $\mathbb{E}\left[Y_{it}|D_{i1},X_{1i1}, u_{i} \right]= \theta_{0} +\gamma t + \beta D_{it} + \theta_{1}X_{1it} + u_{i}$.
	\item[5.] $\hat{\tau}_{IPW}$ - inverse propensity weighted estimate based on post-intervention response (i.e. $Y_{i}=Y_{i1}$) with correctly specified PS model: $\logit\left[ \pi(D_{i1}|X_{1i0},X_{2i},V_{i}) \right]= \alpha_{0} + \alpha_{1}X_{1i0} + \alpha_{2}X_{2i} + \alpha_{3}V_{i}$.
	\item[6.] $\tilde{\tau}_{IPW}$ - inverse propensity weighted estimate based on post-intervention response with incorrectly specified PS model with erroneous exclusion of the time-invariant confounder $X_{2i}$: $\logit\left[ \pi(D_{i1}|X_{1i0},X_{2i},V_{i}) \right]= \alpha_{0} + \alpha_{1}X_{1i0}  + \alpha_{3}V_{i}$.
	\item[7.] $\hat{\tau}_{IPWDID}$ - inverse propensity weighted difference and difference estimate based on correctly specified PS model: $\logit\left[ \pi(D_{i1}|X_{1i0},X_{2i},V_{i}) \right]= \alpha_{0} + \alpha_{1}X_{1i0} + \alpha_{2}X_{2i} + \alpha_{3}V_{i}$.
	\item[8.] $\hat{\tau}_{IPWDID}$ - inverse propensity weighted difference and difference estimate with incorrectly specified PS model with erroneous exclusion of the time-invariant confounder $X_{2i}$: $\logit\left[ \pi(D_{i1}|X_{1i0},X_{2i},V_{i}) \right]= \alpha_{0} + \alpha_{1}X_{1i0}  + \alpha_{3}V_{i}$.
	\item[9.] $\hat{\hat{\tau}}_{DRGLMM}$ - DR estimate based on correctly specified GLMM model and correctly specified PS model.
	\item[10.] $\hat{\tau}_{DRGLMM}$ - DR estimate based on correctly specified GLMM model but augmented with incorrectly specified PS model with erroneous exclusion of the time-invariant confounder $X_{2i}$.
	\item[11.] $\tilde{\tau}_{DRGLMM}$ - DR estimate based on incorrectly specified GLMM with erroneous exclusion of the time-invariant confounder $X_{2i}$ but augmented with correctly specified PS covariates.
	\item[12.] $\hat{\tilde{\tau}}_{DRGLMM}$ - DR estimate based on incorrectly specified GLMM and augmented with incorrectly specified PS covariates where time-invariant confounder $X_{2i}$ is excluded in both the models.
\end{itemize}
The bias, variance (Var) and mean squared error (MSE) of different competitive estimators of the ATE are presented in Table \ref{Tab_1} based on 10000 replications. It is clearly seen that $\tilde{\tau}_{GLMM}$, $\tilde{\tau}_{OR}$ and $\tilde{\tau}_{IPW}$ are biased as the effect of the treatment is counfounded due to the erroneous omission of the time-invariant confounder $X_{2}$. As expected, the performance of $\hat{\tau}_{GLMM}$ is the best with respect to MSE. The DR estimates $\hat{\hat{\tau}}_{DRGLMM}$, $\hat{\tau}_{DRGLMM}$ and $\tilde{\tau}_{DRGLMM}$ perform as well as $\hat{\tau}_{GLMM}$, but the performance of $\hat{\tilde{\tau}}_{DRGLMM}$ is similar to that of $\tilde{\tau}_{GLMM}$ with respect to both bias and variance. 
The bias and variance of $\hat{\tau}_{OR}$ are significantly larger than those of $\hat{\tau}_{GLMM}$ due to the omission of the pre-intervention data. It is also observed that the variance and MSE of $\hat{\tau}_{IPWDID}$ is much less than those of $\hat{\tau}_{IPW}$ for the same reason. This clearly demonstrates that the pre-intervention data contains important information and provides more precise estimates of the ATE under correct model specifications. 

To analyse the sensitivity of the estimators associated with the interaction between time and covariates, we now consider the following relationships between the covariates and the response $Y_{it}$:
\begin{eqnarray}\label{Model-II}
Y_{it}&=& 10 + 3t + 15 D_{it} + X_{1it} + 2tX_{2i} + \log(X_{2i}) + u_{i} + \epsilon_{it}.  \nonumber
\end{eqnarray}
where we have considered an interaction effect of time and time-invariant confounder $X_{2}$ keeping all other specifications as in the previous case. Accordingly, we redefine the aforementioned list of estimators used for comparison.
The results are presented in Table \ref{Tab_2}. With the exception for the case of $\tilde{\tau}_{DRGLMM}$, the performance of all the estimators are similar. Note that, the performance of $\tilde{\tau}_{DRGLMM}$ is better than both $\tilde{\tau}_{GLLM}$ and $\tilde{\tau}_{IPWDID}$, as expected, but it is not as good as $\hat{\tau}_{GLLM}$. Interestingly, $\hat{\tau}_{IPWDID}$ performs better than DR method when the PS model is correctly specified and the GLLM is misspecified due to the erroneous omission of the interaction effect of time with the time-invariant confounder $X_2$.

Next we consider heterogeneous treatment effect using the following relationships between the covariates and the treatment $D_{it}$ and response $Y_{it}$:
\begin{eqnarray}\label{Model-I}
Y_{it}&=& 10 + 2t + 15 D_{it} + X_{1it} + 3X_{2i} + D_{it}X_{1it}+  u_{i} + \epsilon_{it}.  \nonumber
\end{eqnarray}
The bias, variance (Var) and mean squared error (MSE) of different competitive estimators of the ATE and ATT are presented in Tables \ref{Tab_3} and \ref{Tab_4}, respectively. The performance of the estimators for the ATE is similar to the results for homogeneous treatment effects as presented in \ref{Tab_1}. As expected, the erroneous omission of the time-invariant confounder $X_{2}$, results in high bias for $\tilde{\kappa}_{GLMM}$, $\tilde{\kappa}_{OR}$ and $\tilde{\kappa}_{IPW}$. The performance of $\hat{\kappa}_{GLMM}$ is the best and the DR estimates $\hat{\hat{\kappa}}_{DRGLMM}$, $\hat{\kappa}_{DRGLMM}$ and $\tilde{\kappa}_{DRGLMM}$ perform as well as $\hat{\kappa}_{GLMM}$. As expected, the performance of $\hat{\tilde{\kappa}}_{DRGLMM}$ is similar to that of $\tilde{\kappa}_{GLMM}$ with respect to both bias and variance. 
Here also, the bias and variance of $\hat{\kappa}_{OR}$ and $\hat{\kappa}_{IPW}$ are significantly larger than those of $\hat{\kappa}_{GLMM}$ and $\hat{\kappa}_{IPWDID}$, respectively, due to the omission of data. It is also observed that the performance of $\hat{\tau}_{OR}$ and $\hat{\kappa}_{OR}$ further deteriorates (not reported here) with the pre-intervention component $X_{1i0}$ instead of the post-intervention component $X_{1i1}$ in the conditioning set when the distribution of $X_{1it}$ changes with time $t$.

Now, we consider the interaction between time and covariates in presence of heterogeneous treatment effect using the following relationships between the covariates and the response $Y_{it}$:
\begin{eqnarray}\label{Model-I}
Y_{it}&=& 10 + 2t + 15 D_{it} + X_{1it} + 3tX_{2i} + D_{it}X_{1it}+  u_{i} + \epsilon_{it},  \nonumber
\end{eqnarray}
and present the simulation results in Tables \ref{Tab_5}-\ref{Tab_6}. The performance of all the estimators, except $\tilde{\tau}_{DRGLMM}$ and $\tilde{\kappa}_{DRGLMM}$, are similar to the case without time-interaction. Here also, the performances of $\tilde{\tau}_{DRGLMM}$ and $\tilde{\kappa}_{DRGLMM}$ are better than all the estimates based on misspecified models. Similar to the case of homogeneous treatment effects, $\hat{\tau}_{IPWDID}$ performs better than DR method when the PS model is correctly specified and the GLLM is misspecified due to the erroneous omission of the interaction effect of time with the time-invariant confounder $X_2$. A similar pattern is also observed for $\hat{\kappa}_{IPWDID}$. 

To understand the utility of the random effects $u_i$, we have also evaluated the performance of the estimates, ignoring them in model fitting under each of the aforementioned simulation settings (not reported here). As expected, the variance of the estimates increases without accounting for the between-cluster heterogeneity, and that results in higher MSE. It is important to note that the MSE associated with the misspecified models, including DR models, increases drastically in the absence of random effects. 

\begin{table}[!ht]
	\caption{Simulation results for estimation of ATE with homogeneous treatment effects.}
	\centering
	\begin{tabular}{c c c c c c c c c c c c}
		\hline\hline
		&  & & & $n=250$ &  &  \vline  &  &  $n=500$ &   \\ 
		\hline
		Estimator &  OR/ GLLM & PS &  Bias $\times100$ & Var &  MSE & \vline  & Bias $\times100$ & Var &  MSE \\ 
		\hline
		$\hat{\tau}_{OR}$ & Correct & -  & 1.309 & 0.925 & 0.925 & \vline  & 0.425 &  0.452 &  0.452  \\
		$\tilde{\tau}_{OR}$ & Incorrect & -  &  92.999 & 1.335 & 2.200 & \vline  & 93.232 & 0.659 & 1.529  \\
		$\hat{\tau}_{GLMM}$ & Correct & -  & 0.301 & 0.549 & 0.549 & \vline  &  0.436 & 0.274 & 0.274   \\
		$\tilde{\tau}_{GLMM}$ & Incorrect & -   & 23.324 & 0.586 & 0.640 & \vline  & 23.821 & 0.294 & 0.351   \\
		$\hat{\tau}_{IPW}$ & - & Correct  & 7.228 & 2.411 & 2.416 & \vline  &  1.915 & 1.031 & 1.032   \\
		$\tilde{\tau}_{IPW}$ & - & Incorrect  & 96.495 & 2.154 & 3.085 & \vline  & 95.553 & 0.968 & 1.881   \\
		$\hat{\tau}_{IPWDID}$ & - & Correct & 0.754 & 0.830 &  0.830 & \vline  &  0.612 & 0.414 & 0.414   \\
		$\tilde{\tau}_{IPWDID}$ & - & Incorrect & 0.760 & 0.872 &  0.872 & \vline  & 0.711 & 0.430 & 0.430  \\
		$\hat{\hat{\tau}}_{DRGLMM}$ & Correct & Correct  & 0.209 & 0.556 & 0.556 & \vline  & 0.174 & 0.277 & 0.277   \\
		$\hat{\tau}_{DRGLMM}$ & Correct & Incorrect & 0.219 & 0.555 & 0.555 & \vline  & 0.454 & 0.277 & 0.277   \\
		$\tilde{\tau}_{DRGLMM}$ & Incorrect & Correct& 2.0409 & 0.573 & 0.573 & \vline  & 2.034 & 0.287 & 0.287   \\
		$\hat{\tilde{\tau}}_{DRGLMM}$ & Incorrect & Incorrect  & 22.929 & 0.592 & 0.644 & \vline  & 23.776 & 0.297 & 0.353   \\
		\hline
	\end{tabular}
	\begin{tablenotes}
		\small
		\item OR: Outcome regression estimate; GLLM: Generalized liner mixed model estimate;
		IPW: Inverse propensity weighted estimate;
		IPWDID: Inverse propensity weighted DID estimate; DRGLLM: Doubly robust GLLM estimate.
	\end{tablenotes}
	\label{Tab_1}
\end{table}

\begin{table}[!ht]
	\centering
	\caption{Simulation results for estimation of ATE with homogeneous treatment effects and time-interactions.}
	\centering
	\begin{tabular}{c c c c c c c c c c c c}
		\hline\hline
		&  & & & $n=250$ &  &  \vline  &  &  $n=500$ &   \\ 
		\hline
		Estimator &  OR/ GLLM & PS &  Bias $\times100$ & Var &  MSE & \vline  & Bias $\times100$ & Var &  MSE \\ 
		\hline
		$\hat{\tau}_{OR}$ & Correct & -  & 1.334 & 0.923 & 0.923 & \vline  & 0.425 &  0.452 &  0.452  \\
		$\tilde{\tau}_{OR}$ & Incorrect & -  &  93.031 & 1.333 & 2.199 & \vline  & 93.232 & 0.659 & 1.529  \\
		$\hat{\tau}_{GLMM}$ & Correct & -  & 0.357 &  0.553 &  0.553 & \vline  &   0.499 & 0.276 & 0.276   \\
		$\tilde{\tau}_{GLMM}$ & Incorrect & -   & 78.354 &  0.856 & 1.470 & \vline  & 79.100 & 0.424 & 1.050   \\
		$\hat{\tau}_{IPW}$ & - & Correct  & 7.081 & 1.77 & 1.77 & \vline  &  2.030 & 0.809 & 0.809   \\
		$\tilde{\tau}_{IPW}$ & - & Incorrect  & 94.713 & 1.750 & 2.647 & \vline  &  93.797 & 0.828 & 1.708   \\
		$\hat{\tau}_{IPWDID}$ & - & Correct &  3.756 & 0.924 &  0.924 & \vline  & 1.274 & 0.457 & 0.457   \\
		$\tilde{\tau}_{IPWDID}$ & - & Incorrect & 74.243 & 1.062 &  1.613 & \vline  & 74.214 & 0.523 & 1.074  \\
		$\hat{\hat{\tau}}_{DRGLMM}$ & Correct & Correct  &  -0.063 & 0.560 & 0.560 & \vline  &  0.106 & 0.278 & 0.278   \\
		$\hat{\tau}_{DRGLMM}$ & Correct & Incorrect & 0.381 & 0.558 & 0.558 & \vline  &  0.485 & 0.277 & 0.277   \\
		$\tilde{\tau}_{DRGLMM}$ & Incorrect & Correct& 57.961 & 0.745 & 1.081 & \vline  & 58.109 & 0.365 & 0.703   \\
		$\hat{\tilde{\tau}}_{DRGLMM}$ & Incorrect & Incorrect  & 78.461 & 0.861 & 1.477 & \vline  & 79.258 & 0.425 & 1.053   \\
		\hline
	\end{tabular}
	\begin{tablenotes}
		\small
		\item OR: Outcome regression estimate; GLLM: Generalized liner mixed model estimate;
		IPW: Inverse propensity weighted estimate;
		IPWDID: Inverse propensity weighted DID estimate; DRGLLM: Doubly robust GLLM estimate.
	\end{tablenotes}
	\label{Tab_2}
\end{table}

\begin{table}[!ht]
	\caption{Simulation results for estimation of ATE with heterogeneous treatment effects.}
	\centering
	\begin{tabular}{c c c c c c c c c c c c}
		\hline\hline
		&  & & & $n=250$ &  &  \vline  &  &  $n=500$ &   \\
		\hline
		Estimator &  OR/ GLLM & PS &  Bias $\times100$ & Var &  MSE & \vline  & Bias $\times100$ & Var &  MSE \\ 
		\hline
		$\hat{\tau}_{OR}$ & Correct & -  & 0.014 & 0.928 & 0.928 & \vline  & 0.008 &  0.454 &  0.454  \\
		$\tilde{\tau}_{OR}$ & Incorrect & -  &  3.120 & 1.498 & 2.690 & \vline  & 3.140 & 0.741 &  1.949  \\
		$\hat{\tau}_{GLMM}$ & Correct & -  & 0.004 & 0.575 & 0.575 & \vline  &  0.013 & 0.285 & 0.285   \\
		$\tilde{\tau}_{GLMM}$ & Incorrect & -   & 0.707 & 0.622 & 0.682 & \vline  & 0.724 & 0.312 & 0.377   \\
		$\hat{\tau}_{IPW}$ & - & Correct  & 0.248 & 2.824 & 2.831 & \vline  &  0.071 & 1.204 & 1.204   \\
		$\tilde{\tau}_{IPW}$ & - & Incorrect  & 3.276 & 2.520 & 3.834 & \vline  & 3.252 & 1.124 & 2.419   \\
		$\hat{\tau}_{IPWDID}$ & - & Correct & 0.036 & 0.958 &  0.958 & \vline  &  0.029 & 0.465 & 0.465   \\
		$\tilde{\tau}_{IPWDID}$ & - & Incorrect & 0.018 & 0.908 &  0.908 & \vline  & 0.024 & 0.446 & 0.446  \\
		$\hat{\hat{\tau}}_{DRGLMM}$ & Correct & Correct  & 0.001 & 0.582 & 0.582 & \vline  & 0.013 & 0.288 & 0.288   \\
		$\hat{\tau}_{DRGLMM}$ & Correct & Incorrect & 0.001 & 0.582 & 0.582 & \vline  & 0.013 & 0.288 & 0.288   \\
		$\tilde{\tau}_{DRGLMM}$ & Incorrect & Correct& 0.054 & 0.608 & 0.608 & \vline  & 0.050 & 0.304 & 0.304   \\
		$\hat{\tilde{\tau}}_{DRGLMM}$ & Incorrect & Incorrect  & 0.694 &  0.627 & 0.686 & \vline  & 0.722 & 0.315 & 0.379   \\
		\hline
	\end{tabular}
	\begin{tablenotes}
		\small
		\item OR: Outcome regression estimate; GLLM: Generalized liner mixed model estimate;
		IPW: Inverse propensity weighted estimate;
		IPWDID: Inverse propensity weighted DID estimate; DRGLLM: Doubly robust GLLM estimate.
	\end{tablenotes}
	\label{Tab_3}
\end{table}

\begin{table}[!ht]
	\caption{Simulation results for estimation of ATT with heterogeneous treatment effects.}
	\centering
	\begin{tabular}{c c c c c c c c c c c c}
		\hline\hline
		&  & & & $n=250$ &  &  \vline  &  &  $n=500$ &   \\
		\hline
		Estimator &  OR/ GLLM & PS &  Bias $\times100$ & Var &  MSE & \vline  & Bias $\times100$ & Var &  MSE \\ 
		\hline
		$\hat{\kappa}_{OR}$ & Correct & -  & 0.144 & 0.986 & 0.986 & \vline  & -0.570 &  0.501 &  0.501  \\
		$\tilde{\kappa}_{OR}$ & Incorrect & -  &  108.770 & 1.609 & 2.793 & \vline  & 109.566 & 0.806 & 2.007  \\
		$\hat{\kappa}_{GLMM}$ & Correct & -  & -0.089 & 0.585 & 0.585 & \vline  &  -0.588 &  0.294 & 0.294   \\
		$\tilde{\kappa}_{GLMM}$ & Incorrect & -   &   23.945 & 0.638 & 0.695 & \vline  & 24.047 &  0.319 & 0.380   \\
		$\hat{\kappa}_{IPW}$ & - & Correct  & 12.895 & 4.235 & 4.251 & \vline  &  2.728 & 1.845 & 1.845   \\
		$\tilde{\kappa}_{IPW}$ & - & Incorrect  & 114.824 & 2.892 & 4.210 & \vline  & 112.846 & 1.467 &  2.600   \\
		$\hat{\kappa}_{IPWDID}$ & - & Correct & 0.583 &  1.011 &   1.011 & \vline  &  0.612 & 0.492 & 0.492   \\
		$\tilde{\kappa}_{IPWDID}$ & - & Incorrect &  0.034 & 0.953 &  0.953 & \vline  & 0.966 & 0.467 & 0.467  \\
		$\hat{\hat{\kappa}}_{DRGLMM}$ & Correct & Correct  &  0.022 & 0.591 & 0.591 & \vline  & 0.498 &  0.295 &  0.295   \\
		$\hat{\kappa}_{DRGLMM}$ & Correct & Incorrect &   -0.025 & 0.591 & 0.591 & \vline  & 0.653 & 0.307 & 0.307   \\
		$\tilde{\kappa}_{DRGLMM}$ & Incorrect & Correct&  1.625 & 0.622 & 0.622 & \vline  &  2.020 & 0.312 & 0.312   \\
		$\hat{\tilde{\kappa}}_{DRGLMM}$ & Incorrect & Incorrect  & 23.911 &  0.643 & 0.700 & \vline  &  24.85 &  0.320 & 0.382   \\
		\hline
	\end{tabular}
	\begin{tablenotes}
		\small
		\item OR: Outcome regression estimate; GLLM: Generalized liner mixed model estimate;
		IPW: Inverse propensity weighted estimate;
		IPWDID: Inverse propensity weighted DID estimate; DRGLLM: Doubly robust GLLM estimate.
	\end{tablenotes}
	\label{Tab_4}
\end{table}

\begin{table}[!ht]
	\caption{Simulation results for estimation of ATE with time-interactions and heterogeneous treatment effects.}
	\centering
	\begin{tabular}{c c c c c c c c c c c c}
		\hline\hline
		&  & & & $n=250$ &  &  \vline  &  &  $n=500$ &   \\
		\hline
		Estimator &  OR/ GLLM & PS &  Bias $\times100$ & Var &  MSE & \vline  & Bias $\times100$ & Var &  MSE \\ 
		\hline
		$\hat{\tau}_{OR}$ & Correct & -  & 0.479 & 0.928 & 0.928 & \vline  & 0.274 &  0.454 &  0.454  \\
		$\tilde{\tau}_{OR}$ & Incorrect & -  &  109.185 & 1.498 & 2.690 & \vline  & 109.910 & 0.741 &  1.949  \\
		$\hat{\tau}_{GLMM}$ & Correct & -  &  0.185 & 0.580 & 0.580 & \vline  &  0.500 & 0.287 & 0.287   \\
		$\tilde{\tau}_{GLMM}$ & Incorrect & -   & 106.830 & 1.156 & 2.297 & \vline  & 107.862 & 0.573 & 1.736   \\
		$\hat{\tau}_{IPW}$ & - & Correct  & 8.692 & 2.824 & 2.831 & \vline  &  2.492 &  1.204 & 1.204   \\
		$\tilde{\tau}_{IPW}$ & - & Incorrect  & 114.676 & 2.520 & 3.834 & \vline  & 113.820 & 1.124 & 2.419   \\
		$\hat{\tau}_{IPWDID}$ & - & Correct & 4.854 & 1.331 &  1.334 & \vline  &  1.795 &  0.625 &  0.625   \\
		$\tilde{\tau}_{IPWDID}$ & - & Incorrect & 112.586 &  1.601 & 2.869 & \vline  & 113.262 & 0.762 & 2.045  \\
		$\hat{\hat{\tau}}_{DRGLMM}$ & Correct & Correct  & 0.082 & 0.586 & 0.586 & \vline  &  0.502 & 0.290 & 0.290   \\
		$\hat{\tau}_{DRGLMM}$ & Correct & Incorrect & 0.086 & 0.586 & 0.586 & \vline  & 0.517 & 0.290 & 0.290   \\
		$\tilde{\tau}_{DRGLMM}$ & Incorrect & Correct& 82.556 & 0.955 & 1.637 & \vline  & 83.101 & 0.466 & 1.157   \\
		$\hat{\tilde{\tau}}_{DRGLMM}$ & Incorrect & Incorrect  & 107.484 &  1.164 & 2.319 & \vline  & 108.753 &  0.577 & 1.759   \\
		\hline
	\end{tabular}
	\begin{tablenotes}
		\small
		\item OR: Outcome regression estimate; GLLM: Generalized liner mixed model estimate;
		IPW: Inverse propensity weighted estimate;
		IPWDID: Inverse propensity weighted DID estimate; DRGLLM: Doubly robust GLLM estimate.
	\end{tablenotes}
	\label{Tab_5}
\end{table}

\begin{table}[!ht]
	\caption{Simulation results for estimation of ATT with time-interaction and heterogeneous treatment effects.}
	\centering
	\begin{tabular}{c c c c c c c c c c c c}
		\hline\hline
		&  & & & $n=250$ &  &  \vline  &  &  $n=500$ &   \\
		\hline
		Estimator &  OR/ GLLM & PS &  Bias $\times100$ & Var &  MSE & \vline  & Bias $\times100$ & Var &  MSE \\ 
		\hline
		$\hat{\kappa}_{OR}$ & Correct & -  & -0.676 & 1.010 & 1.010 & \vline  &  1.383 &   0.497 &   0.497  \\
		$\tilde{\kappa}_{OR}$ & Incorrect & -  &   107.835 & 1.601 & 2.764 & \vline  & 110.497 & 0.797 &  2.018  \\
		$\hat{\kappa}_{GLMM}$ & Correct & -  & -0.720 & 0.599 & 0.599 & \vline  &  1.477 &  0.294 & 0.294   \\
		$\tilde{\kappa}_{GLMM}$ & Incorrect & -   &   104.276 & 1.158 & 2.246 & \vline  & 106.876 &  0.584 & 1.726   \\
		$\hat{\kappa}_{IPW}$ & - & Correct  & 9.072 & 4.050 & 4.057 & \vline  &  5.828 &  1.722 &  1.725   \\
		$\tilde{\kappa}_{IPW}$ & - & Incorrect  & 112.967 & 2.871 & 4.147 & \vline  & 114.284 & 1.307 &  2.613   \\
		$\hat{\kappa}_{IPWDID}$ & - & Correct & 4.325 &  1.738 &   1.739 & \vline  &   4.859 & 0.795 & 0.797   \\
		$\tilde{\kappa}_{IPWDID}$ & - & Incorrect &  109.767 & 1.662 &  2.867 & \vline  & 114.033 & 0.805 & 2.106  \\
		$\hat{\hat{\kappa}}_{DRGLMM}$ & Correct & Correct  &  -0.734 & 0.605 & 0.605 & \vline  & 1.493 &  0.296 &  0.297   \\
		$\hat{\kappa}_{DRGLMM}$ & Correct & Incorrect &    -0.725 & 0.605 & 0.605 & \vline  & 1.447 &  0.297 &  0.297   \\
		$\tilde{\kappa}_{DRGLMM}$ & Incorrect & Correct&  80.527 & 0.969 & 1.618 & \vline  &  82.923 & 0.485 & 1.173   \\
		$\hat{\tilde{\kappa}}_{DRGLMM}$ & Incorrect & Incorrect  & 105.125 &  1.168 & 2.273 & \vline  &  107.715 &  0.588 & 1.748   \\
		\hline
	\end{tabular}
	\begin{tablenotes}
		\small
		\item OR: Outcome regression estimate; GLLM: Generalized liner mixed model estimate;
		IPW: Inverse propensity weighted estimate;
		IPWDID: Inverse propensity weighted DID estimate; DRGLLM: Doubly robust GLLM estimate.
	\end{tablenotes}
	\label{Tab_6}
\end{table}

\clearpage

\subsection{Sensitivity Analysis in Presence of Unobserved Heterogeneity }\label{Sensitivity}
In order to study the effect of unobserved heterogeneity on the performance of the competitive estimators, we first consider the following model for data generation:
\begin{eqnarray}\label{Model-III}
Y_{it}&=& \theta_{0i} + \gamma t + \beta D_{it} + \theta_{1i}X_{1it} + \theta_{2i}X_{2i} + \theta_{3i}D_{it}X_{1it}+  u_{i} + \epsilon_{it},  \nonumber
\end{eqnarray}
where $\theta_{0i}$, $\theta_{1i}$, $\theta_{2i}$, $\theta_{3i}$, $\gamma_i$ and $\beta_i$ follows independent Gaussian distribution with mean 10, 1, 3, 1, 2, 15, respectively, and standard deviation as 10\% of the corresponding mean value.  The results are presented in the Tables \ref{Tab_7}-\ref{Tab_8}. We also consider the following model:
\begin{eqnarray}\label{Model-III}
Y_{it}&=& \theta_{0i} + \gamma t + \beta D_{it} + \theta_{1i}X_{1it} + \theta_{2i}t X_{2i} + \theta_{3i}D_{it}X_{1it}+  u_{i} + \epsilon_{it},  \nonumber
\end{eqnarray}
where interaction of time $t$ with $X_{2i}$ is considered instead of $X_{2i}$, and the results are presented in Table  \ref{Tab_9}-\ref{Tab_10}. In comparison with the results presented in Tables \ref{Tab_3}-\ref{Tab_6}, both the bias and variance of the estimators are marginally higher due to the presence of unobserved heterogeneity. However, the relative performance of the estimators with respect to mean squared error remains the same. Here also, we observe that the MSE of the estimates increases in the absence of random effects for the unaccounted between-cluster heterogeneity.

\begin{table}[!ht]
	\caption{Simulation results for estimation of ATE with unobserved heterogeneity.}
	\centering
	\begin{tabular}{c c c c c c c c c c c c}
		\hline\hline
		&  & & & $n=250$ &  &  \vline  &  &  $n=500$ &   \\
		\hline
		Estimator &  OR/ GLLM & PS &  Bias $\times100$ & Var &  MSE & \vline  & Bias $\times100$ & Var &  MSE \\ 
		\hline
		$\hat{\tau}_{OR}$ & Correct & -  & -0.891 & 1.073 & 1.073 & \vline  & -1.413 &  0.541 &  0.541  \\
		$\tilde{\tau}_{OR}$ & Incorrect & -  &  109.112 & 1.657 & 2.847 & \vline  & 107.245 & 0.842 &  1.992  \\
		$\hat{\tau}_{GLMM}$ & Correct & -  & -0.586 & 0.644 & 0.644 & \vline  &  -1.393 & 0.324 & 0.324   \\
		$\tilde{\tau}_{GLMM}$ & Incorrect & -   & 24.904 & 0.678 & 0.740 & \vline  & 24.009 & 0.343 & 0.401   \\
		$\hat{\tau}_{IPW}$ & - & Correct  & 5.365 & 3.094 & 3.096 & \vline  &  0.925 & 1.449 & 1.449   \\
		$\tilde{\tau}_{IPW}$ & - & Incorrect  & 113.703 & 2.692 & 3.985 & \vline  & 110.791 & 1.246 & 2.473   \\
		$\hat{\tau}_{IPWDID}$ & - & Correct &  0.410 & 0.950 &  0.950 & \vline  &  -1.004 & 0.492 & 0.492   \\
		$\tilde{\tau}_{IPWDID}$ & - & Incorrect & -0.297 & 1.012 & 1.012 & \vline  & -1.122 & 0.469 & 0.470  \\
		$\hat{\hat{\tau}}_{DRGLMM}$ & Correct & Correct  & -0.766 & 0.651 & 0.651 & \vline  & -1.361 & 0.327 & 0.327   \\
		$\hat{\tau}_{DRGLMM}$ & Correct & Incorrect & -0.793 & 0.651 & 0.651 & \vline  & -1.376 & 0.327 & 0.327   \\
		$\tilde{\tau}_{DRGLMM}$ & Incorrect & Correct& 1.001 & 0.665 & 0.665 & \vline  &  0.013 & 0.331 & 0.331   \\
		$\hat{\tilde{\tau}}_{DRGLMM}$ & Incorrect & Incorrect  & 24.449 &  0.684 & 0.744 & \vline  & 23.993 & 0.345 & 0.403   \\
		\hline
	\end{tabular}
	\begin{tablenotes}
		\small
		\item OR: Outcome regression estimate; GLLM: Generalized liner mixed model estimate;
		IPW: Inverse propensity weighted estimate;
		IPWDID: Inverse propensity weighted DID estimate; DRGLLM: Doubly robust GLLM estimate.
	\end{tablenotes}
	\label{Tab_7}
\end{table}

\begin{table}[!ht]
	\caption{Simulation results for estimation of ATT with unobserved heterogeneity.}
	\centering
	\begin{tabular}{c c c c c c c c c c c c}
		\hline\hline
		&  & & & $n=250$ &  &  \vline  &  &  $n=500$ &   \\
		\hline
		Estimator &  OR/ GLLM & PS &  Bias $\times100$ & Var &  MSE & \vline  & Bias $\times100$ & Var &  MSE \\ 
		\hline
		$\hat{\kappa}_{OR}$ & Correct & -  & 0.616 & 1.195 & 1.195 & \vline  & 0.399 &  0.573 &  0.573  \\
		$\tilde{\kappa}_{OR}$ & Incorrect & -  &  109.886 & 1.824 & 3.027 & \vline  & 109.700 & 0.863 & 2.064  \\
		$\hat{\kappa}_{GLMM}$ & Correct & -  & 1.721 & 0.685 & 0.685 & \vline  &  -0.014 &  0.330 & 0.330   \\
		$\tilde{\kappa}_{GLMM}$ & Incorrect & -   &   26.807 & 0.713 & 0.784 & \vline  &  24.947 &  0.345 & 0.409   \\
		$\hat{\kappa}_{IPW}$ & - & Correct  & 9.727& 4.890 & 4.890 & \vline  &  3.712 & 1.947 & 1.947   \\
		$\tilde{\kappa}_{IPW}$ & - & Incorrect  & 113.146 & 3.136 & 4.412 & \vline  & 112.9 & 1.418 &  2.690   \\
		$\hat{\kappa}_{IPWDID}$ & - & Correct & 2.363 &  1.102 &   1.103 & \vline  &   -0.018 & 0.513 & 0.513   \\
		$\tilde{\kappa}_{IPWDID}$ & - & Incorrect &  1.757 & 1.033 &  1.033 & \vline  & -0.224 & 0.487 & 0.487  \\
		$\hat{\hat{\kappa}}_{DRGLMM}$ & Correct & Correct  &  1.812 & 0.691 & 0.691 & \vline  &  0.015 &  0.332 &  0.332   \\
		$\hat{\kappa}_{DRGLMM}$ & Correct & Incorrect &    1.790 & 0.691 & 0.692 & \vline  &  0.087 & 0.332 & 0.332   \\
		$\tilde{\kappa}_{DRGLMM}$ & Incorrect & Correct&  3.809 & 0.703 & 0.704 & \vline  &  1.349 & 0.341 & 0.341   \\
		$\hat{\tilde{\kappa}}_{DRGLMM}$ & Incorrect & Incorrect  & 26.620 &  0.718 & 0.788 & \vline  &  24.958 &  0.348 & 0.410   \\
		\hline
	\end{tabular}
	\begin{tablenotes}
		\small
		\item OR: Outcome regression estimate; GLLM: Generalized liner mixed model estimate;
		IPW: Inverse propensity weighted estimate;
		IPWDID: Inverse propensity weighted DID estimate; DRGLLM: Doubly robust GLLM estimate.
	\end{tablenotes}
	\label{Tab_8}
\end{table}

\begin{table}[!ht]
	\caption{Simulation results for estimation of ATE with time-interaction and unobserved heterogeneity.}
	\centering
	\begin{tabular}{c c c c c c c c c c c c}
		\hline\hline
		&  & & & $n=250$ &  &  \vline  &  &  $n=500$ &   \\
		\hline
		Estimator &  OR/ GLLM & PS &  Bias $\times100$ & Var &  MSE & \vline  & Bias $\times100$ & Var &  MSE \\ 
		\hline
		$\hat{\tau}_{OR}$ & Correct & -  & -1.365 & 1.081 & 1.081 & \vline  & -0.615 &  0.544 &  0.544  \\
		$\tilde{\tau}_{OR}$ & Incorrect & -  &  108.230 & 1.651 & 2.822 & \vline  & 108.947 & 0.830 &  2.017  \\
		$\hat{\tau}_{GLMM}$ & Correct & -  & -0.075 & 0.653 & 0.653 & \vline  &  -0.604 & 0.327 & 0.327   \\
		$\tilde{\tau}_{GLMM}$ & Incorrect & -   & 107.348 & 1.242 & 2.394 & \vline  & 106.692 & 0.621 & 1.760   \\
		$\hat{\tau}_{IPW}$ & - & Correct  & 4.058 & 2.904 & 2.905 & \vline  &  1.161 & 1.298 & 1.298   \\
		$\tilde{\tau}_{IPW}$ & - & Incorrect  & 113.051 & 2.601 & 3.879 & \vline  & 111.553 & 1.248 & 2.493   \\
		$\hat{\tau}_{IPWDID}$ & - & Correct & 3.209 & 1.392 &  1.393 & \vline  &  0.742 & 0.628 & 0.628   \\
		$\tilde{\tau}_{IPWDID}$ & - & Incorrect & 113.175 & 1.625 & 2.905 & \vline  & 111.632 & 0.793 & 2.039  \\
		$\hat{\hat{\tau}}_{DRGLMM}$ & Correct & Correct  & -0.124 & 0.658 & 0.658 & \vline  & -0.683 & 0.331 & 0.331   \\
		$\hat{\tau}_{DRGLMM}$ & Correct & Incorrect & -0.090 & 0.660 & 0.659 & \vline  & -0.643 & 0.331 & 0.331   \\
		$\tilde{\tau}_{DRGLMM}$ & Incorrect & Correct& 83.290 & 1.038 & 1.737 & \vline  & 82.363 & 0.520 & 1.198   \\
		$\hat{\tilde{\tau}}_{DRGLMM}$ & Incorrect & Incorrect  & 108.199 &  1.252 & 2.423 & \vline  & 107.519 & 0.627 & 1.783   \\
		\hline
	\end{tabular}
	\begin{tablenotes}
		\small
		\item OR: Outcome regression estimate; GLLM: Generalized liner mixed model estimate;
		IPW: Inverse propensity weighted estimate;
		IPWDID: Inverse propensity weighted DID estimate; DRGLLM: Doubly robust GLLM estimate.
	\end{tablenotes}
	\label{Tab_9}
\end{table}

\begin{table}[!ht]
	\caption{Simulation results for estimation of ATT with time-interaction and unobserved heterogeneity.}
	\centering
	\begin{tabular}{c c c c c c c c c c c c}
		\hline\hline
		&  & & & $n=250$ &  &  \vline  &  &  $n=500$ &   \\
		\hline
		Estimator &  OR/ GLLM & PS &  Bias $\times100$ & Var &  MSE & \vline  & Bias $\times100$ & Var &  MSE \\ 
		\hline
		$\hat{\kappa}_{OR}$ & Correct & -  &  -0.269 & 1.151 & 1.151 & \vline  & 0.362 &  0.567 &  0.567  \\
		$\tilde{\kappa}_{OR}$ & Incorrect & -  &  109.024 & 1.751 & 2.938 & \vline  & 109.822 & 0.867 &  2.071  \\
		$\hat{\kappa}_{GLMM}$ & Correct & -  &  -0.456 & 0.671 & 0.671 & \vline  &  0.251 & 0.324 & 0.324   \\
		$\tilde{\kappa}_{GLMM}$ & Incorrect & -   &  105.590 & 1.257 & 2.370 & \vline  & 106.059 & 0.616 & 1.739   \\
		$\hat{\kappa}_{IPW}$ & - & Correct  & 10.907 & 4.393 & 4.405 & \vline  &  5.382 & 1.841 & 1.844   \\
		$\tilde{\kappa}_{IPW}$ & - & Incorrect  &  115.473 & 2.988 & 4.319 & \vline  & 114.356 & 1.383 & 2.689   \\
		$\hat{\kappa}_{IPWDID}$ & - & Correct &  5.222 & 1.887 &  1.890 & \vline  &  2.935 & 0.831 & 0.832   \\
		$\tilde{\kappa}_{IPWDID}$ & - & Incorrect & 111.832 & 1.751 & 2.300 & \vline  & 112.449 & 0.824 & 2.086  \\
		$\hat{\hat{\kappa}}_{DRGLMM}$ & Correct & Correct  & -0.546 & 0.677 & 0.677 & \vline  & 0.189 & 0.327 & 0.327   \\
		$\hat{\kappa}_{DRGLMM}$ & Correct & Incorrect & -0.469 & 0.677 & 0.677 & \vline  &  0.221 & 0.327 & 0.327   \\
		$\tilde{\kappa}_{DRGLMM}$ & Incorrect & Correct& 82.230 & 1.070 & 1.744 & \vline  & 82.488 & 0.522 & 1.201   \\
		$\hat{\tilde{\kappa}}_{DRGLMM}$ & Incorrect & Incorrect  & 106.241 &  1.267 & 2.394 & \vline  & 106.825 & 0.620 & 1.759   \\
		\hline
	\end{tabular}
	\begin{tablenotes}
		\small
		\item OR: Outcome regression estimate; GLLM: Generalized liner mixed model estimate;
		IPW: Inverse propensity weighted estimate;
		IPWDID: Inverse propensity weighted DID estimate; DRGLLM: Doubly robust GLLM estimate.
	\end{tablenotes}
	\label{Tab_10}
\end{table}

\section{Quantifying Causal Effects of Cycle Superhighways on Traffic Volume and Speed}\label{Data}
As discussed in Section \ref{intro}, we analyse the causal effect of CS on traffic congestion based on a dataset collected over the period 2007-2014. In this study, 75 treated zones and 375 control zones were selected using stratified random sampling along the 40 km long main corridors radiating from central London to outer London such that interference among the treated and control units is minimal.

In observational data the effect of CS is confounded due to various factors related to traffic dynamics, road characteristics, and socio-demographic conditions. The traffic data on the major road network as well as on the minor road network are collected by The Department for Transport \citep{DFT}. Also, additional information on traffic flow and speed are collected from the London Atmospheric Emissions Inventory (LAEI). 
It is observed that traffic congestion is associated with bus-stop density and road network density \cite{Bus}. An association between traffic congestion and socio-demographic characteristics, such as employment and land-use patterns, has also been indicated in previous studies \citep{Land1, Land2, Employ}. To incorporate these effects, we obtained relevant data on population and employment density, as well as the information of land-use patterns from the Office for National Statistics. The traffic characteristics are time-varying but the road characteristics, land-use patterns and employment density remain unchanged over the period of our study. The data that are available from the aforementioned sources and the logic to construct responses and covariates are described below.
\begin{itemize}
	\item[(a)] Annual average daily traffic (AADT) -- the total volume of vehicle traffic of a highway or road for a year divided by 365 days. To measure AADT on individual road segments, traffic data is collected by an automated traffic counter, hiring an observer to record traffic or licensing estimated counts from GPS data providers. AADT is a simple, but useful, measurement to indicate busyness of a road.
	\item[(b)] Traffic speed -- calculated using time-mean-speed method based on the individual speed records for vehicles passing a point over a selected time period. Speed is also a fundamental measurement in transport engineering and used for maintaining a designated level of service.
	\item[(c)] Total Cycle Collisions (TCS) -- total number of injured cycle collisions based on police records from the STATS19 accident reporting form and collected by the UK Department for Transport. The location of an accident is recorded using coordinates which are in accordance with the British National Grid coordinate system. The CS routes were intended to reduce the risk of accidents for cyclists and the route allocation is possibly influenced by TCS. But, previous studies indicate that CS routes are not more dangerous or safer than the control roads \citep{Safety}. It is expected that the accident rates will affect the traffic characteristics.
	\item[(d)] Bus-stop density -- the ratio of the number of bus-stops to the road length. The presence of bus-stops is expected to affect the traffic flow and speed due to frequent bus-stops and pedestrian activities.  The allocation of CS routes were designed to avoid areas with high bus-stop density for safety of the cyclists.
	\item[(e)] Road network density -- with the available geographical information system we could also represent the road network density in each zone by using a measure of the number of network nodes per unit of area. A network node is defined as the meeting point of two or more links. To safeguard from conflicting turning movements the CS paths are routed through the areas with high road network density.  
	\item[(f)] Road length -- high capacity networks tend to depress land values which in turn will influence the socio-economic profile of the people who live close together. Data for road length for each zone was generated using geographical information system software.
	\item[(g)] Road type -- a binary variable where `1' represents dual-carriageway and `0' represents single-carriage. This is an important feature since we might expect traffic congestion in single-carriage roads.
	\item[(h)] Density of domestic buildings -- this is a potentially useful feature since we might expect congestion to be associated with the nature of land use and the degree of urbanization. Also, the allocation of the CS paths are possibly influenced by land use characteristics.
	\item[(i)] Density of non-domestic buildings -- rising housing costs in business and office districts force people to live further away, lengthening commutes, and affecting traffic flow and speed. As mentioned before, this feature may influence allocation of the CS paths.
	\item[(j)] Road area density -- the ratio of the area of the zone's total road network to the land area of the zone. The road network includes all roads in the including motorways, highways, main or national roads, secondary or regional roads. It is expected that the traffic flow is associated with road density. 
	\item[(k)] Employment density -- traffic generation potential depends on economic activity and we proxy this by employment density. High employment density tends to influence pedestrian activity which in turn affects traffic speed. The CS paths are designed to provide coverage in the areas with high employment density and encourage commuters to use cycling as a regular mode of transport. 
	\item[(l)] Time -- in a longitudinal study, time itself may be a confounder because government policies or other interventions could simultaneously affect AADT and speed. Fuel taxation or motoring policies, for instance, provide relevant examples.
\end{itemize}

First, we analyse the causal effect of CS on AADT considering the sources of confounding  as mentioned in (c)-(l). We consider a generalized additive model to estimate the PS based on pre-intervention measurements. We used backward elimination following \citet{Beng} for covariate selection to avoid issues related to multicollinearity. At each step, the covariate with the largest p-value was dropped one at a time until all covariates are significant with a cut-point of p-value = 0.10. In this process, the final PS model include the following factors: TCS, bus-stop density, road  network  density, road length, density  of  domestic  buildings, density of non-domestic buildings, road area density, and employment density. To test our PS specification we check for balancing. In a similar manner to \citet{Diag1} and \citet{Diag2}, we regress the exposure using a GAM on the covariates, and the estimated PS up to a cubic term. 
The adjusted $R^2$ value obtained from the GAM without covariate is greater than that of with covariates.
This result suggests that the balancing property has been achieved for our PS specification as the inclusion of covariates leads to a deterioration in model adequacy.

To model AADT we consider a Gaussian GLMM and found that BIC values support identity rather than log link function. The GLMM model with gamma family encounters convergence issues. A completely analogous algorithm is used for variable selection and the final model included exposure along with density of domestic  buildings, density of non-domestic  buildings, road area density, road network density, road length, road type, bus-stop density and its interaction with the exposure, time and its interaction with density of non-domestic  buildings. Similarly, we built a DR model using the approach described in Subsection \ref{DR}. For this augmented GLMM model we consider four dummy variables based on equally distributed classification of the estimated PS. The models are estimated by maximum likelihood using the {\ttfamily{nlme}} package in R. Some elements of $\boldsymbol{\hat{\zeta}}$ are significantly different from zero, which indicates some deficiency in the GLMM model. To investigate further, we conduct a formal diagnostic tests using the approach proposed by \citet{Diag3}. We compute empirical variances $\sigma^{2}_{DRI}$  and $\sigma^{2}_{DRG}$
of $\hat{\tau}_{DRGLLM} - \hat{\tau}_{IPWDID}$ and $\hat{\tau}_{DRGLLM} - \hat{\tau}_{GLLM}$, respectively, based on 10000 nonparametric bootstrap replications. Then we consider the test statistics $|\frac{\hat{\tau}_{DRGLLM} - \hat{\tau}_{IPWDID}}{\sigma_{DRI}}|$ and $|\frac{\hat{\tau}_{DRGLLM} - \hat{\tau}_{GLLM}}{\sigma_{DRG}}|$ with critical region 1.96 (for level 0.05) to test the the null hypotheses that the PS model and the GLLM model, respectively, are correctly specified. Similar tests are also conducted based on the estimates of ATT. We do not find any evidence of misspecification for both the PS and GLLM model.

The estimated ATE and ATT of CS on AADT relative to the average AADT in the pre-intervention period are presented in Table \ref{Result_AADT}. The standard errors (SE) and the corresponding 95\% confidence intervals (CI)  for all the estimates are obtained from 10000 nonparametric bootstrap samples. The results indicate a reduction in traffic flow compared to the pre-intervention period. The 95\% bootstrap CI exclude $0$, suggesting a significant effect of CS on AADT. The estimates of ATE and ATT vary within $-9.48\%$ to $-6.47\%$, and $-7.03\%$ to $-6.54\%$, respectively, but the GLMM method provides results similar to the DR method. Interestingly, the overall reduction in traffic flow in London is slightly more than those of the treated locations based on the estimates from GLLM and DR methods. The SE of $\hat{\tau}_{IPWDID}$ is higher than those of $\hat{\tau}_{GLMM}$ and  $\hat{\tau}_{DRGLMM}$, which is also reflected in the width of the confidence intervals. Similar patterns are observed for the estimates of ATT. The bootstrap mean of all the estimators is similar to the estimates and that indicates a small bias if the underlying modeling assumptions hold.

\begin{table}[!ht]
	\centering
	\caption{Effect of CS on AADT relative to the average AADT in the pre-intervention period}
	\resizebox{\columnwidth}{2 cm}{
		\begin{tabular}{c c c c c c c}
			\hline\hline
			Estimand &  Estimator &	 Estimate (\%) & Mean & SE & $95\%$ CI \\ 
			\hline
			&	$\hat{\tau}_{IPWDID}$    &   -6.467 & -6.827 & 3.119 & (-14.141,  -1.875)   \\
			ATE	&	$\hat{\tau}_{GLMM}$  & -9.371 & -9.739 & 1.769 &    (-13.205, -6.334)   \\
			&	$\hat{\tau}_{DRGLMM}$ &  -9.479 &  -9.785 & 1.772 &  (-13.261,  -6.370)  \\
			\hline
			&	$\hat{\kappa}_{IPWDID}$    &  -6.536 & -6.452 & 1.482 & (-9.261,  -3.487)   \\
			ATT &		$\hat{\kappa}_{GLMM}$  & -6.950 & -7.302 & 1.195 & (-9.652, -4.986)   \\
			&		$\hat{\kappa}_{DRGLMM}$ &  -7.026 &  -7.334 & 1.197 &  (-9.684,  -5.007)  \\
			\hline
	\end{tabular}}
	\begin{tablenotes}
		\small
		\item IPWDID: Inverse propensity weighted DID estimate;
		\item GLLM: generalized liner mixed model; DRGLLM: Doubly robust GLLM estimate.
	\end{tablenotes}
	\label{Result_AADT}
\end{table}

Next, we perform a similar analysis to estimate the ATE and ATT of CS on traffic speed considering the sources of confounding  as mentioned in (c)-(l). We model speed using a Gaussian GLMM with identity link function and the final model included exposure along with TCS, density of domestic buildings, density of non-domestic buildings, bus-stop density, road type, road length and its interaction with the exposure, time and its interaction with density of domestic buildings. Here also, we do not find any evidence of misspecification for the PS model, but some deficiency in the GLMM model is indicated by the diagnostic tests. 

The estimated ATE and ATT of CS on traffic speed relative to the average speed in the pre-intervention period are presented in Table \ref{Result_speed}. It is observed that the estimates of the change in traffic speed based on all the methods are insignificant. It is not unexpected that we do not observe changes in traffic speed, there are many factors associated with the transport network and other interventions can play a crucial role in mitigating potential traffic problems anticipated by the introduction of cycle lanes. Important factors contributing to the levels of congestion are traffic speed and flow. Considering both the analyses presented here, it can be inferred that Cycle Superhighways are an effective intervention that can potentially alleviate traffic congestion in London.

\begin{table}[!ht]
	\centering
	\caption{Effect of CS on speed relative to the average speed in the pre-intervention period}
	\resizebox{\columnwidth}{2 cm}{
		\begin{tabular}{c c c c c c}
			\hline\hline
			Estimand &  Estimator &	 Estimate (\%) & Mean & SE & $95\%$ CI \\
			\hline
			&	$\hat{\tau}_{IPWDID}$    &   -0.423 & -0.157 &  2.294 & (-3.840,  5.082)   \\
			ATE	&	$\hat{\tau}_{GLMM}$  & 1.787 & 1.416 & 1.230 &    (-0.895,  3.919)   \\
			&	$\hat{\tau}_{DRGLMM}$ & 1.640 &  1.330 & 1.234 &  (-0.990, 3.834)  \\
			\hline
			&	$\hat{\kappa}_{IPWDID}$    &  1.025 & 1.025 & 1.923 & (-3.028,  4.577)   \\
			ATT &		$\hat{\kappa}_{GLMM}$  & 2.245 & 1.821 & 1.399 & (-0.713, 4.795)   \\
			&		$\hat{\kappa}_{DRGLMM}$ & 2.096 &  1.736 & 1.403 &  (-0.813, 4.704)  \\
			\hline
	\end{tabular}}
	\begin{tablenotes}
		\small
		\item IPWDID: Inverse propensity weighted DID estimate;
		GLLM: generalized liner mixed model estimate; DRGLLM: Doubly robust GLLM estimate.
	\end{tablenotes}
	\label{Result_speed}
\end{table}

\section{Concluding Remarks}\label{Conc}
This paper has presented a statistical framework that can be used to derive inference for causal quantities based on pre-intervention and post-intervention data motivated from a case study on the London transport network. However, the scope of the proposed methods go far beyond this particular application. The key methodological insight is that extending the traditional OR model within a GLMM set-up, which is able to represent both time-varying and time-invariant confounding and accounts for the serial correlations in the data. The Inverse propensity weighted difference in difference estimate is an attractive alternative because it avoids any parametric assumptions associated with the form of the regression and/or link function which is essential in the GLMM approach. The proposed DR approach works if either of the GLMM or PS model is correct. This method can be easily extended for the cases with observations on multiple time periods. Our results suggest that the introduction of Cycle Superhighways can reduce traffic flow, but we find marginal improvement in traffic speed. Providing evidence that Cycle Superhighways can be an effective intervention in metropolitan cities like London, which are heavily affected by congestion. It is also worth noting that we are obliged to assume that the outcomes of one unit are not affected by the treatment assignment of any other units. While not always plausible, we tried to reduce the “spillover” effects by reducing interactions between the treated and control
units. One possible direction of further studies could be using an improved design and define several types of treatment effects following \citet{Hudgens_2008, Tchetgen_2012, Brian_2020}, and develop associated estimation methodologies for the setting where there may be clustered interference. In line with our application, we have also assumed that the time-varying confounders are not affected by treatment. However, this may not be a realistic assumption in many scenarios, and one may consider marginal structural models for a valid inference on causal quantities \citep{MSM,Daniel}. In the context of high-dimensional data with a large number of potentially confounding variables, standard model selection techniques are cumbersome and may not provide valid inferences about causal quantities \citep{Belloni_2014,Tan}. It could be another interesting problem to extend the proposed inferential methods in a high-dimensional setting under a sparse structure as considered in \citet{Belloni_2017} and \citet{Tan}.

In recent years, London's air quality has improved as a result of policies to reduce emissions, primarily from road transport, although significant areas still exceed NO2 EU limits. The Cycle Superhighways are one of the several interventions introduced which may results in an improvement in air quality. However, the effect of Cycle Superhighways on air quality is debatable, and is an interesting research problem that could be studied under the same causal analysis set-up outlined here.

\section*{\small Acknowledgement}
The authors would like to acknowledge the Lloyd’s Register Foundation for funding this research through the programme on Data-Centric Engineering at the Alan Turing Institute.

\begin{appendix}
	\section*{}
	\begin{theorem}
		The inverse propensity weighted difference-in-difference estimator $\hat{\tau}_{IPWDID}$ is a consistent estimator of $\tau_{0}$ when the propensity score model $\pi(D_{i1}|\boldsymbol{X_{i0}}, \hat{\alpha})$ is correctly specified, and $\mathbb{E}[Y_{i0}(1)]=\mathbb{E}[Y_{i0}(0)]$.
	\end{theorem}
	
	\begin{proof} 
		Following \citet{Dav}, we consider the following estimating equations
		\sloppy
		\begin{eqnarray}
		\sum_{i=1}^{n}\dfrac{D_{i1}-\pi(D_{i1}|\boldsymbol{X_{i0}},\hat{\alpha})}{\pi(D_{i1}|\boldsymbol{X_{i0}},\hat{\alpha})\left[1-\pi(D_{i1}|\boldsymbol{X_{i0}},\hat{\alpha}\right)]}\frac{\partial }{\partial \hat{\alpha}}\pi(D_{i1}|\boldsymbol{X_{i0}},\hat{\alpha})=0,
		\label{ML}
		\end{eqnarray}
		\begin{eqnarray}
		\sum_{i=1}^{n}\left[\dfrac{I(D_{i1})Y_{it}}{\pi(D_{i1}|\boldsymbol{X_{i0}},\hat{\alpha})}- \hat{\Delta}_{t}^{(1)}\right]=0,
		\label{Del1}
		\end{eqnarray}
		and
		\begin{eqnarray}
		\sum_{i=1}^{n}\left[\dfrac{\left[1-I(D_{i1})\right]Y_{it}}{1-\pi(D_{i1}|\boldsymbol{X_{i0}},\hat{\alpha})}- \hat{\Delta}_{t}^{(0)}\right]=0.
		\label{Del2}
		\end{eqnarray}
		The maximum likelihood estimate $\hat{\alpha}$, $\hat{\Delta}_{t}^{(1)}=\displaystyle\frac{1}{n}\sum_{i=1}^{n}\dfrac{I(D_{i1})Y_{it}}{\pi(D_{i1}|\boldsymbol{X_{i0}},\hat{\alpha})}$, and $\hat{\Delta}_{t}^{(0)}=\displaystyle\frac{1}{n}\sum_{i=1}^{n}\dfrac{\left[1-I(D_{i1})\right]Y_{it}}{1-\pi(D_{i1}|\boldsymbol{X_{i0}},\hat{\alpha})}$ satisfies equation \ref{ML}-\ref{Del2}, for $t=0,1$. 
		
		Note that $(Y_{it}(1),Y_{it}(0)) \independent I(D_{i1})| (X_{i0}, X_{i1})$, and $I(D_{i1})Y_{it}=I(D_{i1})\left[ I(D_{i1})Y_{it}(1) + (1-I(D_{i1}))Y_{it}(0) \right]=I(D_{i1})Y_{it}(1)$ for $t=0,1$. Now we can write 
		\begin{eqnarray}
		\mathbb{E}\left[\frac{I(D_{i1})Y_{it}}{\pi(D_{i1}|\boldsymbol{X_{i0}},\alpha_{0})}\right]&=&\mathbb{E}\left[\frac{I(D_{i1})Y_{it}(1)}{\pi(D_{i1}|\boldsymbol{X_{i0}},\alpha_{0})}\right]\nonumber\\
		&=& \mathbb{E}\left[\mathbb{E}\left\{\frac{I(D_{i1})Y_{it}(1)}{\pi(D_{i1}|\boldsymbol{X_{i0}},\alpha_{0})}\middle| Y_{it}(1), X_{i0}, X_{i1}\right\} \right] \nonumber\\
		&=&\mathbb{E}\left[\frac{Y_{it}(1)}{\pi(D_{i1}|\boldsymbol{X_{i0}},\alpha_{0})}\mathbb{E}\left(I(D_{i1})\middle| Y_{it}(1), X_{i0}, X_{i1}\right) \right]\nonumber\\
		&=&\mathbb{E}\left[\frac{Y_{it}(1)}{\pi(D_{i1}|\boldsymbol{X_{i0}},\alpha_{0})}\mathbb{E}\left(I(D_{i1})\middle|  X_{i0}\right) \right]\nonumber\\
		&=& \mathbb{E}\left[\frac{Y_{it}(1)}{\pi(D_{i1}|\boldsymbol{X_{i0}},\alpha_{0})}\pi(D_{i1}|\boldsymbol{X_{i0}},\alpha_{0})\right]\nonumber\\
		&=&\mathbb{E}\left[Y_{it}(1)\right], \nonumber
		\end{eqnarray}
		where $\alpha_{0}$ is the ture value of $\alpha$.
		Similarly, we have $\mathbb{E}\left[\frac{(1-I(D_{i1}))Y_{it}}{1-\pi(D_{i1}|\boldsymbol{X_{i0}},\alpha_{0})}\right]=\mathbb{E}\left[Y_{it}(0)\right]$ for $t=0,1$. Therefore, 
		\begin{eqnarray}
		\mathbb{E}\left[\dfrac{D_{i1}-\pi(D_{i1}|\boldsymbol{X_{i0}},\alpha_{0})}{\pi(D_{i1}|\boldsymbol{X_{i0}},\alpha_{0)}\left[1-\pi(D_{i1}|\boldsymbol{X_{i0}},\alpha_{0}\right)]}\frac{\partial }{\partial \hat{\alpha}}\pi(D_{i1}|\boldsymbol{X_{i0}},\alpha_{0})\right]=0 \nonumber,
		\end{eqnarray}
		\begin{eqnarray}
		\mathbb{E}\left[\dfrac{I(D_{i1})Y_{it}}{\pi(D_{i1}|\boldsymbol{X_{i0}},\alpha_{0})}- \Delta_{t}^{(1)}\right]=0 \nonumber,
		\end{eqnarray}
		and
		\begin{eqnarray}
		\mathbb{E}\left[\dfrac{\left[1-I(D_{i1})\right]Y_{it}}{1-\pi(D_{i1}|\boldsymbol{X_{i0}},\alpha_{0})}- \Delta_{t}^{(0)}\right]=0 \nonumber,
		\end{eqnarray}
		where $\Delta_{t}^{(0)}=\mathbb{E}\left[Y_{it}(0)\right]$ and $\Delta_{t}^{(1)}=\mathbb{E}\left[Y_{it}(1)\right]$ for $t=0,1$. Note that $\mathbb{E}[Y_{i0}(1)]=\mathbb{E}[Y_{i0}(0)]$ by assumption. Now using consistency of the M-estimates \citep[p-249]{Huber,Serf}, we have
		
		\begin{eqnarray}
		\hat{\tau}_{IPWDID}&=&\left[\hat{\Delta}_{1}^{(1)}-\hat{\Delta}_{1}^{(0)}\right]-\left[\hat{\Delta}_{0}^{(1)}-\hat{\Delta}_{0}^{(0)}\right]\nonumber\\
		&\overset{a.s.}{\longrightarrow}&\Big[\mathbb{E}\left[Y_{i1}(1)\right]-\mathbb{E}\left[Y_{i1}(0)\right]\Big]-\Big[\mathbb{E}\left[Y_{i0}(1)\right]-\mathbb{E}\left[Y_{i0}(0)\right]\Big]\nonumber\\
		&=&\mathbb{E}\left[Y_{i1}(1)\right]-\mathbb{E}\left[Y_{i1}(0)\right]=\tau_{0}.\nonumber
		\end{eqnarray}
		
	\end{proof}

	\begin{theorem}
		The inverse propensity weighted difference-in-difference estimator $\hat{\kappa}_{IPWDID}$ is a consistent estimator of $\kappa_{0}$ when the propensity score model $\pi(D_{i1}|\boldsymbol{X_{i0}}, \hat{\alpha})$ is correctly specified, and $\mathbb{E}[Y_{i0}(1)|D_{i1}=1]=\mathbb{E}[Y_{i0}(0)|D_{i1}=1]$.
	\end{theorem}

	\begin{proof}
		Following \citet{Erica}, we have
		\begin{eqnarray}
		\frac{1}{\sum_{i=1}^{n}D_{i1}}\sum_{i=1}^{n}\left[I(D_{i1})-\frac{\left[1-I(D_{i1})\right]\pi(D_{i1}|\boldsymbol{X_{i0}},\hat{\alpha})}{1-\pi(D_{i1}|\boldsymbol{X_{i0}}, \hat{\alpha})}\right]Y_{i1}\overset{p}{\longrightarrow} \mathbb{E}\left[Y_{i1}(1)|D_{i1}=1\right]-\mathbb{E}\left[Y_{i1}(0)|D_{i1}=1\right]\nonumber.
		\end{eqnarray}
		Similarly, we also have
		\begin{eqnarray}
		\frac{1}{\sum_{i=1}^{n}D_{i1}}\sum_{i=1}^{n}\left[I(D_{i1})-\frac{\left[1-I(D_{i1})\right]\pi(D_{i1}|\boldsymbol{X_{i0}},\hat{\alpha})}{1-\pi(D_{i1}|\boldsymbol{X_{i0}}, \hat{\alpha})}\right]Y_{i0}\overset{p}{\longrightarrow} \mathbb{E}\left[Y_{i0}(1)|D_{i1}=1\right]-\mathbb{E}\left[Y_{i0}(0)|D_{i1}=1\right]\nonumber.
		\end{eqnarray}
		Note that $\mathbb{E}[Y_{i0}(1)|D_{i1}=1]=\mathbb{E}[Y_{i0}(0)|D_{i1}=1]$ by assumption. Hence, ${\hat{\kappa}_{IPWDID}\overset{p}{\longrightarrow} \kappa_{0}}$.
	\end{proof}

\end{appendix}
	 	
	\bibliographystyle{apalike}
	\bibliography{Bibliography}

\begin{thebibliography}{}

\bibitem[Abadie, 2005]{Parallel}
Abadie, A. (2005).
\newblock Semiparametric difference-in-differences estimators.
\newblock {\em The Review of Economic Studies}, 72(1):1--19.

\bibitem[Abadie and Imbens, 2006]{Abadie_2006}
Abadie, A. and Imbens, G.~W. (2006).
\newblock Large sample properties of matching estimators for average treatment
  effects.
\newblock {\em Econometrica}, 74(1):235--267.

\bibitem[Austin, 2011]{PS}
Austin, P.~C. (2011).
\newblock An introduction to propensity score methods for reducing the effects
  of confounding in observational studies.
\newblock {\em Multivariate Behavioral Research}, 46(3):399--424.

\bibitem[Badoe and Miller, 2000]{Land1}
Badoe, D.~A. and Miller, E.~J. (2000).
\newblock Transportation-land-use interaction: empirical findings in north
  america, and their implications for modeling.
\newblock {\em Transportation Research Part D: Transport and Environment},
  5(4):235--263.

\bibitem[Bang and Robins, 2005]{Beng}
Bang, H. and Robins, J.~M. (2005).
\newblock Doubly robust estimation in missing data and causal inference models.
\newblock {\em Biometrics}, 61:962--972.

\bibitem[Barkley et~al., 2020]{Brian_2020}
Barkley, B.~G., Hudgens, M.~G., Clemens, J.~D., Ali, M., and Emch, E.~E.
  (2020).
\newblock Causal inference from observational studies with clustered
  interference, with application to a cholera vaccine study.
\newblock {\em Annals of Applied Statistics}, 14(3):1432--1448.

\bibitem[Belloni et~al., 2017]{Belloni_2017}
Belloni, A., Chernozhukov, V., Fernández‐Val, I., and Hansen, C. (2017).
\newblock Program evaluation and causal inference with high‐dimensional data.
\newblock {\em Econometrica}, 85(1):233--298.

\bibitem[Belloni et~al., 2014]{Belloni_2014}
Belloni, A., Chernozhukov, V., and Hansen, C. (2014).
\newblock High-dimensional methods and inference on structural and treatment
  effects.
\newblock {\em Journal of Economic Perspectives}, 28(2):29--50.

\bibitem[Blunden, 2016]{Evening}
Blunden, M. (2016).
\newblock Cycle superhighways make traffic worse in the city, report reveals.
\newblock {\em EveningStandard}, Oct 5.

\bibitem[Branas et~al., 2011]{Epi}
Branas, C.~C., Cheney, R.~A., MacDonald, J.~M., Tam, V.~W., Jackson, T.~D., and
  Ten~Have, T.~R. (2011).
\newblock A difference-in-differences analysis of health, safety, and greening
  vacant urban space.
\newblock {\em American Journal of Epidemiology}, 174(11):1296–1306.

\bibitem[Callaway and Sant'Anna, 2020]{Callaway}
Callaway, B. and Sant'Anna, P. H.~C. (2020).
\newblock Difference-in-differences with multiple time periods.
\newblock {\em Journal of Econometrics}.

\bibitem[Card and Krueger, 1994]{Market}
Card, D. and Krueger, A.~B. (1994).
\newblock Minimum wages and employment: A case study of the fast food industry
  in new jersey and pennsylvania.
\newblock {\em American Economic Review}, 84(4):772--93.

\bibitem[Chaisemartin and D'Haultfoeuille, 2020]{Chaisemartin_2020}
Chaisemartin, D. and D'Haultfoeuille, X. (2020).
\newblock Two-way fixed effects estimators with heterogeneous treatment
  effects.
\newblock {\em American Economic Review, Available at
  https://arxiv.org/abs/1803.08807}.

\bibitem[Chen et~al., 2001]{Employ}
Chen, C., Jia, Z., and Varaiya, P. (2001).
\newblock Causes and cures of highway congestion.
\newblock {\em IEEE Control Systems Magazine}, 21(6):26--32.

\bibitem[Daniel et~al., 2013]{Daniel}
Daniel, R.~M., Cousens, S.~N., De~Stavola, B.~L., Kenward, M.~G., and Sterne,
  J. A.~C. (2013).
\newblock Methods for dealing with time-dependent confounding.
\newblock {\em Statistics in Medicine}, 32(9):1584--1618.

\bibitem[Daw and Hatfield, 2018]{Daw_2018}
Daw, J.~R. and Hatfield, L.~A. (2018).
\newblock Matching and regression to the mean in difference-in-differences
  analysis.
\newblock {\em Health Services Research}, 53(6):4138--4156.

\bibitem[Diggle et~al., 2004]{Diggle}
Diggle, P.~J., Heagerty, P., Liang, K., and Zeger, S. (2004).
\newblock {\em Analysis of longitudinal data}.
\newblock Oxford University Press Inc.

\bibitem[Flores et~al., 2012]{Diag1}
Flores, C.~A., Flores-Lagunes, A., Gonzalez, A., and Neumann, T.~C. (2012).
\newblock Estimating the effects of length of exposure to instruction in a
  training program: The case of job corps.
\newblock {\em The Review of Economics and Statistics}, 94(1):153--171.

\bibitem[George, 1970]{Bus}
George, K.~A. (1970).
\newblock Transportation compatible land uses and bus-stop location.
\newblock {\em Transactions on The Built Environment}, 44.

\bibitem[Graham et~al., 2016]{Diag2}
Graham, D.~J., McCoy, E.~J., and Stephens, D.~A. (2016).
\newblock Approximate bayesian inference for doubly robust estimation.
\newblock {\em Bayesian Analysis}, 1(1):47--69.

\bibitem[Hahn, 2003]{Hahn_1998}
Hahn, J. (2003).
\newblock On the role of the propensity score in efficient semiparametric
  estimation of average treatment effects.
\newblock {\em Econometrica}, 66.

\bibitem[Harman et~al., 2011]{Health1}
Harman, J.~S., Lemak, C.~H., Al-Amin, M., Hall, A.~G., and Duncan, R.~P.
  (2011).
\newblock Changes in per member per month expenditures after implementation of
  florida’s medicaid reform demonstration.
\newblock {\em Health Services Research}, 43(3):787--804.

\bibitem[Heckman et~al., 1998]{DID}
Heckman, J., Ichimura, H., Smith, J., and Todd, P. (1998).
\newblock Characterizing selection bias using experimental data.
\newblock {\em Econometrica}, 66(5):1017--1098.

\bibitem[Hern\'an and Robins, 2020]{HernanRobins}
Hern\'an, M.~A. and Robins, J. (2020).
\newblock {\em Causal Inference: What If.}
\newblock Chapman and Hall/CRC.

\bibitem[Hirano et~al., 2003]{Hirano_2003}
Hirano, K., Imbens, G.~W., and Ridder, G. (2003).
\newblock Efficient estimation of average treatment effects using the estimated
  propensity score.
\newblock {\em Econometrica}, 71(4).

\bibitem[Horvitz and Thompson, 1952]{HT}
Horvitz, D.~G. and Thompson, D.~J. (1952).
\newblock A generalization of sampling without replacement from a finite
  universe.
\newblock {\em Journal of the American Statistical Association}, 47:663--685.

\bibitem[Huber, 1967]{Huber}
Huber, P.~J. (1967).
\newblock The behavior of maximum likelihood estimates under nonstandard
  conditions.
\newblock {\em Proceedings of the 5th Berkeley Symposium}, 1(373):221--233.

\bibitem[Hudgens and Halloran, 2018]{Hudgens_2008}
Hudgens, M.~G. and Halloran, M.~E. (2018).
\newblock Toward causal inference with interference.
\newblock {\em Journal of the American Statistical Association},
  103(482):832--842.

\bibitem[Jin and Rafferty, 2017]{Growth}
Jin, J. and Rafferty, P. (2017).
\newblock Does congestion negatively affect income growth and employment
  growth? empirical evidence from us metropolitan regions.
\newblock {\em Transport Policy}, 55:1--8.

\bibitem[Kang and Schafer, 2007]{DR}
Kang, J. D.~Y. and Schafer, J.~L. (2007).
\newblock Demystifying double robustness: A comparison of alternative
  strategies for estimating a population mean from incomplete data.
\newblock {\em Statistical}, 22(4):523–539.

\bibitem[Lechner, 0010]{Lechner_2010}
Lechner, L. (20010).
\newblock The estimation of causal effects by difference-in-difference methods.
\newblock {\em Foundations and Trends in Econometrics}, 4(3):165--224.

\bibitem[Lee et~al., 2010]{Lee_2010}
Lee, B.~K., Lessler, J., and Stuart, E.~A. (2010).
\newblock Improving propensity score weighting using machine learning.
\newblock {\em Statistics in Medicine}, 29(3):337--346.

\bibitem[Li et~al., 2017]{Safety}
Li, H., Graham, D.~J., and Liu, P. (2017).
\newblock Safety effects of the london cycle superhighways on cycle collisions.
\newblock {\em Accident Analysis \& Prevention}, 90(A):90--101.

\bibitem[Lindner and McConnell, 2018]{Lindner_2018}
Lindner, S. and McConnell, K.~J. (2018).
\newblock Difference‑in‑differences and matching on outcomes: a tale of two
  unobservables.
\newblock {\em Health Services and Outcomes Research Methodology},
  19:127–144.

\bibitem[Lunceford and Davidian, 2004]{Dav}
Lunceford, J.~K. and Davidian, M. (2004).
\newblock Stratification and weighting via the propensity score in estimation
  of causal treatment effects: a comparative study.
\newblock {\em Statistics in Medicine}, 23:2937--2960.

\bibitem[McCulloch and Neuhaus, 2011]{McCulloch_2011}
McCulloch, C.~E. and Neuhaus, J.~M. (2011).
\newblock Misspecifying the shape of a random effects distribution: Why getting
  it wrong may not matter.
\newblock {\em Statistical Science}, 26(3):388–402.

\bibitem[Moodie et~al., 2018]{Erica}
Moodie, E. E.~M., Saarela, O., and Stephens, D.~A. (2018).
\newblock A doubly robust weighting estimator of the average treatment effect
  on the treated.
\newblock {\em Stat}, 7(1).

\bibitem[Norman, 2017]{Guardian}
Norman, W. (2017).
\newblock Bike lanes don't clog up our roads, they keep london moving.
\newblock {\em The Gaurdian}, Dec 1.

\bibitem[Retallack and Ostendorf, 2019]{Link1}
Retallack, A.~E. and Ostendorf, B. (2019).
\newblock Current understanding of the effects of congestion on traffic
  accidents.
\newblock {\em International Journal of Environmental Research and Public
  Health}, 16(18).

\bibitem[Robins et~al., 2000]{MSM}
Robins, J.~M., Hern\'an, M.~A., and Brumback, B. (2000).
\newblock Marginal structural models and causal inference in epidemiology.
\newblock {\em Epidemiology}, 11(5):550--560.

\bibitem[Robins and Rotnitzky, 2001]{Diag3}
Robins, J.~M. and Rotnitzky, A. (2001).
\newblock A comment on ``inference for semiparametric models: some questions
  and an answer''.
\newblock {\em Statistica Sinica}, 11:920--936.

\bibitem[Rosenbaum and Rubin, 1983]{Bum}
Rosenbaum, P. and Rubin, D.~B. (1983).
\newblock The central role of the propensity score in observational studies for
  causal effect.
\newblock {\em Biometrika}, 40:41--55.

\bibitem[Rubin, 1978]{Rubin}
Rubin, D.~B. (1978).
\newblock Bayesian inference for causal effects: The role of randomization.
\newblock {\em The Annals of Statistics}, 6:34--58.

\bibitem[Sant'Anna and Zhao, 2020]{SantAnna_2020}
Sant'Anna, P. H.~C. and Zhao, J.~B. (2020).
\newblock Doubly robust difference-in-differences estimators.
\newblock {\em Journal of Econometrics}.

\bibitem[Scharfstein et~al., 1999]{Rot}
Scharfstein, D.~O., Rotnitzky, A., and Robins, J.~M. (1999).
\newblock Adjusting for nonignorable drop-out using semiparametric nonresponse
  models.
\newblock {\em Journal of American Statistical Association}, 94:1096--1120.

\bibitem[Serfling, 1980]{Serf}
Serfling, R.~J. (1980).
\newblock {\em Approximation Theorems of Mathematical Statistics}.
\newblock John Wiley \& Sons.

\bibitem[Skrondal and Rabe-Hesketh, 2009]{RE}
Skrondal, A. and Rabe-Hesketh, S. (2009).
\newblock Prediction in multilevel generalized linear models.
\newblock {\em Journal of Royal Statistical Society, Series A},
  172(3):659–687.

\bibitem[Slawson, 2017]{Cost}
Slawson, N. (2017).
\newblock Traffic jams on major uk roads cost economy around £9bn.
\newblock {\em The Gaurdian}, Oct 18.

\bibitem[Tan, 2020]{Tan}
Tan, Z. (2020).
\newblock Model-assisted inference for treatment effects using regularized
  calibrated estimation with high-dimensional data.
\newblock {\em The Annals of Statistics}, 48(2):811--837.

\bibitem[Tchetgen~Tchetgen and VanderWeele, 2012]{Tchetgen_2012}
Tchetgen~Tchetgen, E.~J. and VanderWeele, T.~J. (2012).
\newblock On causal inference in the presence of interference.
\newblock {\em Statistical Methods in Medical Research}, 21:55--75.

\bibitem[{The Department of Transport}, 2018]{DFT}
{The Department of Transport} (2018).
\newblock Traffic statistics-methodology review-alternative data sources.

\bibitem[{Transport for London}, 2010]{Mayor}
{Transport for London} (2010).
\newblock Cycling revolution london.

\bibitem[{Transport for London}, 2011]{TFL11}
{Transport for London} (2011).
\newblock Barclays cycle superhighways evaluation of pilot routes 3 and 7.

\bibitem[{Transport for London}, 2014]{TFL}
{Transport for London} (2014).
\newblock Number of daily cycle journeys in london.

\bibitem[Westreich et~al., 2010]{Westreich_2010}
Westreich, D., Lessler, J., and Funk, M.~J. (2010).
\newblock Propensity score estimation: neural networks, support vector
  machines, decision trees (cart), and meta-classifiers as alternatives to
  logistic regression.
\newblock {\em Journal of Clinical Epidemiology}, 63:826--833.

\bibitem[Wharam et~al., 2007]{Health2}
Wharam, J.~F., Landon, B.~E., Galbraith, A.~A., Kleinman, K.~P., Soumerai,
  S.~B., and Ross-Degnan, D. (2007).
\newblock Emergency department use and subsequent hospitalizations among
  members of a high-deductible health plan.
\newblock {\em JAMA}, 297(10):1093--102.

\bibitem[Yuan et~al., 2015]{Link2}
Yuan, K., Knoop, V.~L., and Hoogendoorn, S.~P. (2015).
\newblock Capacity drop: Relationship between speed in congestion and the queue
  discharge rate.
\newblock {\em Transportation Research Record}, 2491(1):72–80.

\bibitem[Zhang et~al., 2017]{Land2}
Zhang, K., Sun, D.~J., Shen, S., and Zhu, Y. (2017).
\newblock Analyzing spatiotemporal congestion pattern on urban roads based on
  taxi gps data.
\newblock {\em Journal of Transport and Land Use}, 10(1):675--694.

\end{thebibliography}

\end{document}